\documentclass[journal,12pt,draftcls,onecolumn]{IEEEtran}

\usepackage{epsfig}
\usepackage{subfigure}
\usepackage{graphicx}

\usepackage{amsmath}
\usepackage{amssymb}

\def\BibTeX{{\rm B\kern-.05em{\sc i\kern-.025em b}\kern-.08em
    T\kern-.1667em\lower.7ex\hbox{E}\kern-.125emX}}

\newtheorem{claim}{Claim}
\newtheorem{corollary}{Corollary}
\newtheorem{theorem}{Theorem}

\newtheorem{lemma}{Lemma}

\newtheorem{remark}{Remark}

\title{Feedback Capacity of the Gaussian Interference Channel to within 2 Bits}

\author{Changho Suh and David Tse \\
Wireless Foundations\\
University of California at Berkeley \\
Email: \{chsuh, dtse\}@eecs.berkeley.edu
}

\begin{document}

\IEEEoverridecommandlockouts

\IEEEaftertitletext{
\begin{abstract}
We characterize the capacity region to within 2 bits/s/Hz and the symmetric capacity to within 1 bit/s/Hz for the two-user Gaussian interference channel (IC) with feedback. We develop achievable schemes and derive a new outer bound to arrive at this conclusion. One consequence of the result is that feedback provides \emph{multiplicative} gain, i.e., the gain becomes arbitrarily large for certain channel parameters. It is a surprising result because feedback has been so far known to provide no gain in memoryless point-to-point channels and only \emph{bounded additive} gain in multiple access channels. The gain comes from using feedback to maximize resource utilization, thereby enabling more efficient resource sharing between the interfering users. The result makes use of a deterministic model to provide insights into the Gaussian channel. This deterministic model is a special case of El Gamal-Costa deterministic model and as a side-generalization, we establish the exact feedback capacity region of this general class of deterministic ICs.
\end{abstract}
}

\maketitle

\section{Introduction}

Shannon showed that feedback does not increase capacity in memoryless point-to-point channels \cite{shannon:it}. On the other hand,  feedback can indeed increase capacity in memory channels such as colored Gaussian noise channels. While it can provide multiplicative gain especially in the low $\sf SNR$ regime, the gain is \emph{bounded}, i.e., feedback can provide a capacity increase of at most one bit \cite{Cover:it89,Kim:it06}. In the multiple access channel (MAC), however, Gaarder and Wolf \cite{Gaarder:it} showed that feedback could increase capacity although the channel is memoryless. Inspired by this result, Ozarow \cite{Ozarow:it} found the feedback capacity region for the two-user Gaussian MAC. Ozarow's result reveals that feedback provides only \emph{additive power} gain. The reason for the bounded power gain is that in the MAC, transmitters cooperation induced by feedback can at most boost signal power via aligning signal directions. Boosting signal power provides a capacity increase of a constant number of bits.

A question arises: will feedback help significantly in other channels where a receiver wants to decode only desired message in the presence of interference? To answer this question, we focus on the simple two-user Gaussian interference channel (IC) where each receiver wants to decode the message only from its corresponding transmitter. We first make progress on the symmetric capacity. Gaining insights from a deterministic model \cite{Salman:allterton07} and Alamouti's scheme~\cite{Alamouti:jsac98}, we develop a simple two-staged achievable scheme. We then derive a new outer bound to show that the proposed scheme achieves the symmetric capacity to within one bit for all values of the channel parameters.

An interesting consequence of this result is that feedback can provide \emph{multiplicative} gain in interference channels. This can be shown from the generalized degrees-of-freedom (g.d.o.f.) in Fig.~\ref{fig:gdof}. The notion was defined in \cite{dtse:it07} as
\begin{align}
d(\alpha) \triangleq \lim_{\mathsf{SNR}, \mathsf{INR} \rightarrow \infty} \frac{C_{\mathsf{sym}}(\mathsf{SNR},\mathsf{INR})}{ \log \mathsf{SNR}},
\end{align}
where $\alpha$ ($x$-axis) indicates the ratio of $\mathsf{INR}$ to $\mathsf{SNR}$ in dB scale: $\alpha \triangleq \frac{\log \mathsf{INR}}{ \log \mathsf{SNR}}$.
Notice that in certain weak interference regimes ($0 \leq \alpha \leq \frac{2}{3}$) and in the very strong interference regime ($ \alpha \geq 2$), feedback gain becomes arbitrarily large as $\mathsf{SNR}$ and $\mathsf{INR}$ go to infinity. 
\begin{figure}[h]
\begin{center}
{\epsfig{figure=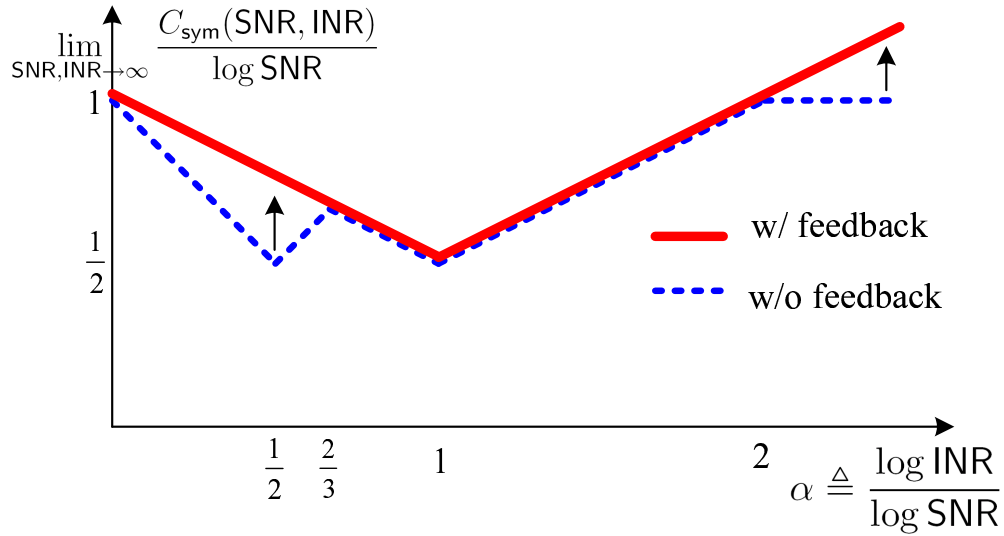, angle=0, width=0.8\textwidth}}
\end{center}
\caption{The generalized degrees-of-freedom of the Gaussian interference channel (IC) with feedback. For certain weak interference regimes ($0 \leq \alpha \leq \frac{2}{3}$) and for the very strong interference regime ($ \alpha \geq 2$), the gap between the non-feedback and the feedback capacity becomes arbitrarily large as $\sf SNR$ and $\sf INR$ go to infinity. This implies that feedback can provide unbounded gain.} \label{fig:gdof}
\end{figure}
For instance, when $\alpha =\frac{1}{2}$, the gap between the non-feedback and the feedback capacity becomes unbounded with the increase of $\sf SNR$ and $\sf INR$, i.e.,
\begin{align}
C_{\sf sym}^{\sf FB} - C_{\sf sym}^{\sf NO} \longrightarrow  \frac{1}{4} \log {\sf SNR} \longrightarrow \infty.
\end{align}
Observing the ratio of the feedback to the non-feedback capacity in the high $\sf SNR$ regime, one can see that feedback provides \emph{multiplicative} gain (50\% gain for $\alpha = \frac{1}{2}$): $\frac{C_{\sf sym}^{\sf FB}}{C_{\sf sym}^{\sf NO}} \rightarrow 1.5$.

Moreover, we generalize the result to characterize the feedback capacity region to within 2 bits per user for all values of the channel parameters. Unlike the symmetric case, we develop an infinite-staged achievable scheme that employs three techniques: (i) block Markov encoding~\cite{Cover:it79,Cover:it81}; (ii) backward decoding~\cite{Kuhlmann:it89}; and (iii) Han-Kobayashi message splitting~\cite{HanKoba:it81}. This result shows interesting contrast with the non-feedback capacity result. In the non-feedback case, it has been shown that the capacity region is described by five types of inequalities including the bounds for $R_1 + 2R_2$ and $2R_1 + R_2$~\cite{HanKoba:it81,dtse:it07}. On the other hand, our result shows that the feedback capacity region requires only three types of inequalities without the $R_1 + 2R_2$ and $2R_1 + R_2$ bounds.


We also develop new interpretation, what we call a \emph{resource hole interpretation}, to provide qualitative insights as to where feedback gain comes from. We find that the gain comes from using feedback to maximize resource utilization, thereby enabling more efficient resource sharing between the interfering users. Also the efficient resource utilization due to feedback turns out to deactivate the $2R_1 + R_2$ bound.

Our results make use of a deterministic model~\cite{Salman:allterton07} to provide insights into the Gaussian channel. This deterministic model is a special case of El Gamal-Costa model~\cite{ElGamal:it82}. As a side-generalization, we establish the exact feedback capacity region of this general class of deterministic ICs. From this result, one can infer an approximate feedback capacity region of two-user Gaussian MIMO ICs, as Teletar and Tse~\cite{Teletar:isit07} did in the non-feedback case.

Interference channels with feedback have received previous attention~\cite{Kramer:it02, Kramer:it04, GastparKramer:06, Jiang:07}. Kramer~\cite{Kramer:it02,Kramer:it04} developed a feedback strategy in the Gaussian IC; Kramer and Gastpar~\cite{GastparKramer:06} derived an outer bound. However, the gap between the inner and outer bounds becomes arbitrarily large with the increase of $\sf SNR$ and $\sf INR$.\footnote{Although this strategy can be arbitrarily far from optimality, a careful analysis reveals that it can also provide multiplicative feedback gain. See Fig. \ref{fig:gdof-compare} for this.} Jiang-Xin-Garg~\cite{Jiang:07} found an achievable region in the discrete memoryless IC with feedback, based on block Markov encoding~\cite{Cover:it79} and binning. However, their scheme involves three auxiliary random variables and therefore requires further optimization. Also no outer bounds are provided. We propose explicit achievable schemes and derive a new tighter outer bound to characterize the capacity region to within 2 bits and the symmetric capacity to within 1 bit universally. Subsequent to our work, Prabhakaran and Viswanath~\cite{Vinod:arix09} have found an interesting connection between our feedback problem and the conferencing encoder problem. Making such a connection, they have independently characterized the sum feedback capacity to within 19 bits/s/Hz. 

\section{Model}
\label{sec-DIC}

Fig. \ref{fig:Gaussian} describes the two-user Gaussian IC with feedback.
Without loss of generality, we normalize signal power and noise power to 1, i.e., $P_k=1$, $Z_k \sim \mathcal{CN}(0,1)$, $\forall k=1,2$. Hence, the signal-to-noise ratio and the interference-to-noise ratio can be defined to capture channel gains:
\begin{align}
\begin{split}
 \mathsf{SNR}_1 &\triangleq |g_{11}|^2, \; \mathsf{SNR}_2 \triangleq |g_{22}|^2, \\
 \mathsf{INR}_{12} &\triangleq |g_{12}|^2, \; \mathsf{INR}_{21} \triangleq |g_{21}|^2.
\end{split}
\end{align}
There are two independent and uniformly distributed sources, $W_k \in \left\{ 1,2,\cdots, m_k \right\}, \forall k=1,2$.
Due to feedback, the encoded signal $X_{ki}$ of user $k$ at time $i$ is a function of its own message and past output sequences:
\begin{align}
X_{ki} = f_{k}^{i} \left(W_k,Y_{k1},\cdots,Y_{k(i-1)} \right)= f_k^{i} \left(W_k,Y_{k}^{i-1} \right),
\end{align}
where we use shorthand notation $Y_{k}^{i-1}$ to indicate the sequence up to $i-1$.
A rate pair $(R_1,R_2)$ is achievable if there exists a family of codebook pairs with codewords (satisfying power constraints) and decoding functions such that the average decoding error probabilities go to zero as block length $N$  goes to infinity. The capacity region $\mathcal{C}$ is the closure of the set of the achievable rate pairs.


\begin{figure}[t]
\begin{center}
{\epsfig{figure=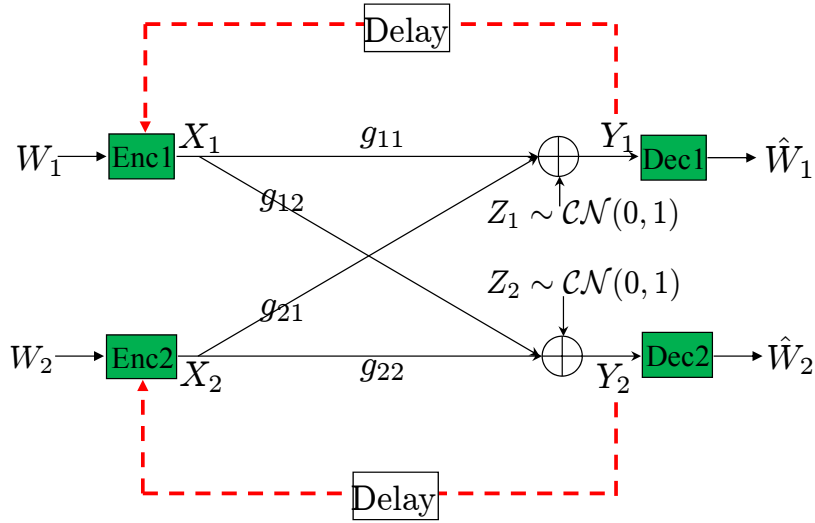, angle=0, width=0.65\textwidth}}
\end{center}
\caption{The Gaussian interference channel (IC) with feedback} \label{fig:Gaussian}
\end{figure}



\section{Symmetric Capacity to within One Bit}
\label{sec-Symmetric}

We start with a symmetric channel setting where $|g_{11}|=|g_{22}|=|g_d|$ and $|g_{12}|=|g_{21}|=|g_c|$:
\begin{align}
\begin{split}
\mathsf{SNR} \triangleq {\sf SNR}_1 = {\sf SNR}_2, \; \mathsf{INR} \triangleq {\sf INR}_{12} = {\sf INR}_{21}.
\end{split}
\end{align}
Not only is this symmetric case simple, it also provides the key ingredients to both the achievable scheme and outer bound needed for the characterization of the capacity region. Furthermore, this case provides enough qualitative insights as to where feedback gain comes from. Hence, we first focus on the symmetric channel.
Keep in mind however that our proposed scheme for a symmetric rate is different from that for a rate region: in the symmetric case, the scheme employs only two stages (or blocks), while an infinite number of stages are used in the general case. We will address this in Section~\ref{sec:FB_CapacityRegion}.

The symmetric capacity is defined by
\begin{align}
C_{\mathsf{sym}} = \sup \left\{R: (R,R) \in \mathcal{C} \right\},
\end{align}
where $\mathcal{C}$ is the capacity region.

\begin{figure}[t]
\begin{center}
{\epsfig{figure=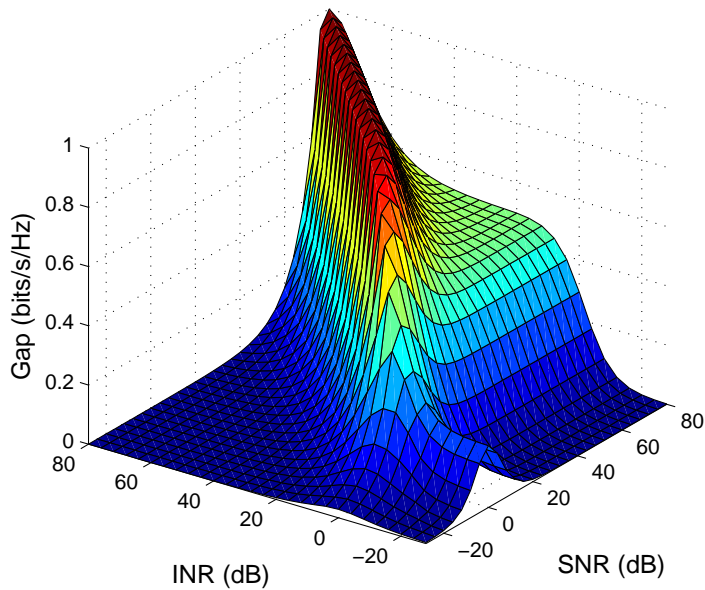, angle=0, width=0.6\textwidth}}
\end{center}
\caption{The gap between our inner and upper bounds. The gap is upper-bounded by exactly one bit. The worst-case gap occurs when ${\sf SNR } \approx {\sf INR}$ and these values go to infinity. In the strong interference regime, the gap vanishes with the increase of $\sf SNR$ and $\sf INR$, while in the weak interference regime, the gap does not, e.g., the gap is around 0.5 bits for $\alpha = \frac{1}{2}$.} \label{fig:Gap_Proposed}
\end{figure}

\begin{theorem}
\label{theorem-symmetric}
We can achieve a symmetric rate of
\begin{align}
\begin{split}
\label{eq:SymmetricAchievableRate}
R_{\sf sym} =  \max \left( \frac{1}{2} \log \left(1 + \mathsf{INR} \right),
\frac{1}{2} \log \left( \frac{ (1+ \sf SNR + \sf INR)^2 - \frac{\sf SNR}{1 + \sf INR} }{1 + 2 \sf INR} \right) \right).
\end{split}
\end{align}
The symmetric capacity is upper-bounded by
\begin{align}
\begin{split}
\label{eq:SymmetricUpperBound}
\overline{C}_{\mathsf{sym}} = \frac{1}{2} \sup_{0 \leq \rho \leq 1} \left[ \log \left(1 + \frac{(1-\rho^2) \mathsf{SNR}}{ 1 + (1-\rho^2)\mathsf{INR}} \right) + \log \left( 1 + \mathsf{SNR} + \mathsf{INR} + 2 \rho \sqrt{\mathsf{SNR} \cdot \mathsf{INR} } \right) \right].
\end{split}
\end{align}
For all channel parameters of $\mathsf{SNR}$ and $\mathsf{INR}$, we can achieve all rates $R$ up to $\overline{C}_{\sf sym}-1$, i.e.,
\begin{align}
\overline{C}_{\mathsf{sym}} -1 \leq R \leq  \overline{C}_{\sf sym}.
\end{align}
\end{theorem}

\begin{figure}[t]
\begin{center}
{\epsfig{figure=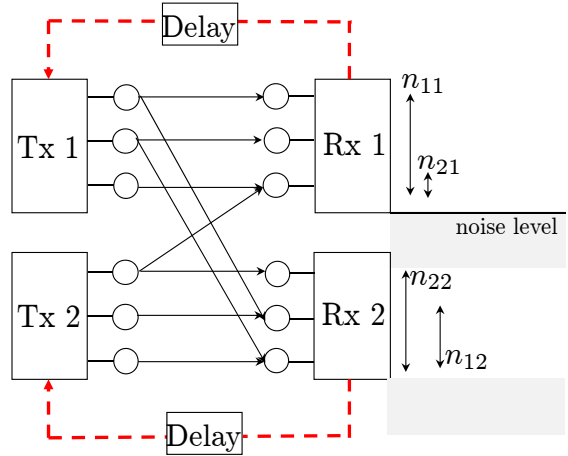, angle=0, width=0.45\textwidth}}
\end{center}
\caption{The deterministic IC with feedback} \label{fig:DIC}
\end{figure}


\subsection{Deterministic Model}
As a stepping stone towards the Gaussian IC, we use an intermediate model: a linear deterministic model~\cite{Salman:allterton07}, illustrated in Fig.~\ref{fig:DIC}. This model is useful in the non-feedback Gaussian IC: it was shown in~\cite{bresler:europe} that the deterministic IC can approximate the Gaussian IC to within a constant number of bits irrespective of the channel parameter values. Our approach is to first develop insights from this model and then translate them to the Gaussian channel.

The connection with the Gaussian channel is as follows. The deterministic IC is characterized by four values: $n_{11}, n_{12}, n_{21}$ and $n_{22}$ where $n_{ij}$ indicates the number of signal bit levels (or resource levels) from transmitter $i$ to receiver $j$.
These values correspond to channel gains in dB scale, i.e., $\forall i \neq j$,
\begin{align}
\label{eq:DIC_GIC_Connection}
n_{ii} =  \lfloor  \log  \mathsf{SNR}_{i} \rfloor, \;n_{ij} =  \lfloor  \log \mathsf{INR}_{ij} \rfloor.
\end{align}
In the symmetric channel, $n \triangleq n_{11}=n_{22}$ and $m \triangleq n_{12} = n_{21}$.
Signal bit levels at a receiver can be mapped to the binary streams of a signal above the noise level. Upper signal levels correspond to most significant bits and lower signal levels correspond to least significant bits. A signal bit level observed by both the receivers above the noise level is broadcasted.
If multiple signal levels arrive at the same signal level at a receiver, we assume a modulo-2-addition.

\subsection{Achievable Scheme for the Deterministic IC}

\textbf{Strong Interference Regime ($m \geq n$):}
\begin{figure}[t]
\begin{center}
{\epsfig{figure=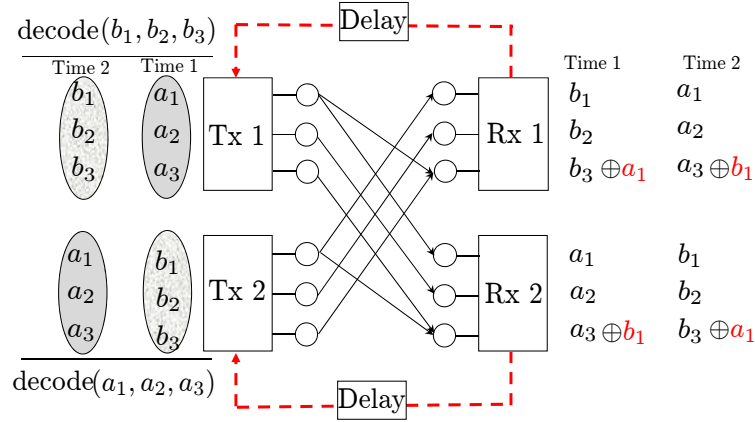, angle=0, width=0.6\textwidth}}
\end{center}
\caption{An achievable scheme for the deterministic IC: strong interference regime $\alpha:=\frac{m}{n}=3$.} \label{fig:Symmetric_Strong_DIC}
\end{figure}
We explain the scheme through the simple example of $\alpha:=\frac{m}{n}=3$, illustrated in Fig.~\ref{fig:Symmetric_Strong_DIC}. Note that each receiver can see only one signal level from its corresponding transmitter. Therefore, in the non-feedback case, each transmitter can send only 1 bit through the top signal level. However, feedback can create a better alternative path, e.g., $[transmitter 1 \rightarrow receiver 2 \rightarrow feedback \rightarrow transmitter 2 \rightarrow receiver 1]$. This alternative path enables to increase the non-feedback rate.

The feedback scheme consists of two stages. In the first stage, transmitters 1 and 2 send independent binary symbols $(a_1, a_2, a_3)$ and $(b_1,b_2,b_3)$, respectively. Each receiver defers decoding to the second stage. In the second stage, using feedback, each transmitter decodes information of the other user: transmitters 1 and 2 decode $(b_1,b_2,b_3)$ and $(a_1, a_2, a_3)$, respectively.
Each transmitter then sends the other user's information. Each receiver gathers the received bits sent during the two stages: the six linearly independent equations containing the six unknown symbols. As a result, each receiver can solve the linear equations to decode its desired bits. Notice that the second stage was used for refining all the bits sent previously, without sending additional information. Therefore, the symmetric rate is $\frac{3}{2}$ in this example. Notice the 50\% improvement from the non-feedback rate of 1. We can easily extend the scheme to arbitrary $(n,m)$. In the first stage, each
transmitter sends $m$ bits using all the signal levels. Using two stages,
these $m$ bits can be decoded with the help of feedback.
Thus, we can achieve:
\begin{align}
\label{eq:DICSymmStrong}
R_{\sf sym} = \frac{m}{2}.
\end{align}

\begin{remark}
The gain in the strong interference regime comes from the fact that feedback provides a better alternative path through the two cross links. The cross links relay the other user's information through feedback. We can also explain this gain using a resource hole interpretation. Notice that in the non-feedback case, each transmitter can send only 1 bit through the top level and therefore there is a resource hole (in the second level) at each receiver. However, with feedback, all of the resource levels at the two receivers can be filled up. Feedback maximizes resource utilization by providing a better alternative path. This concept coincides with correlation routing in~\cite{Kramer:it02}.
\end{remark}

On the other hand, in the weak interference regime, there is no better alternative path, since the cross links are weaker than the direct links. Nevertheless, it turns out that feedback gain can also be obtained in this regime.

\textbf{Weak Interference Regime ($m \leq n$):}
Let us start by examining the scheme in the non-feedback case. Unlike the strong interference regime, only part of information is visible to the other receiver in the weak interference regime. Hence, information can be split into two parts~\cite{HanKoba:it81}: common $m$ bits  (visible to the other receiver) and private $(n-m)$ bits (invisible to the other receiver). Notice that using common levels causes interference to the other receiver.
Sending 1 bit through a common level consumes a total of 2 levels at the two receivers (say \$2), while using a private level costs only \$1. Because of this, a reasonable achievable scheme is to follow the two steps sequentially: (i) sending all of the cheap $(n-m)$ private bits on the lower levels; (ii) sending some number of common bits on the upper levels. The number of common bits is decided depending on $m$ and $n$.

Consider the simple example of $\alpha = \frac{m}{n}= \frac{1}{2}$, illustrated in Fig.~\ref{fig:Symmetric_Weak} $(a)$. First transmitters 1 and 2 use the cheap private signal levels, respectively. Once the bottom levels are used, however using the top levels is precluded due to a conflict with the private bits already sent, thus each transmitter can send only one bit.

\begin{figure}[t]
\begin{center}
{\epsfig{figure=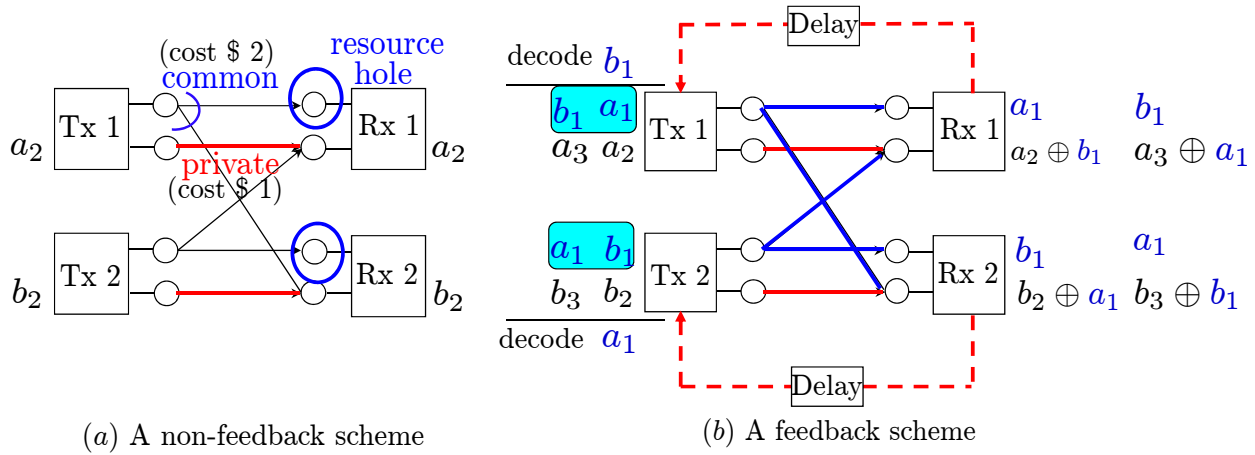, angle=0, width=1.0\textwidth}}
\end{center}
\caption{Achievable schemes for the weak interference regime, e.g., $\alpha= \frac{m}{n}=\frac{1}{2}$.} \label{fig:Symmetric_Weak}
\end{figure}

Observe the two resource holes on the top levels at the two receivers. We find that feedback helps fill up all of these resource holes to improve the performance. The scheme uses two stages. As for the private levels, the same procedure is applied as that in the non-feedback case. How to use the common levels is key to the scheme. 
In the first stage, transmitters 1 and 2 send private bits $a_2$ and $b_2$ on the bottom levels, respectively. Now transmitter 1 squeezes one more bit $a_1$ on its top level. While $a_1$ is received cleanly at receiver 1, it causes interference at receiver 2. Feedback can however resolve this conflict. In the second stage, with feedback transmitter 2 can decode the common bit $a_1$ of the other user.
As for the bottom levels, transmitters 1 and 2 send new private bits $a_3$ and $b_3$, respectively. The idea now is that transmitter 2 sends the other user's common bit $a_1$ on its top level. This transmission allows receiver 2 to refine the corrupted bit $b_2$ from $b_2 \oplus a_1$ without causing interference to receiver 1, since receiver 1 already had the side information of $a_1$ from the previous broadcasting. We paid \$2 for the earlier transmission of $a_1$, but now we can get a \emph{rebate} of \$1. Similarly, with feedback, transmitter 2 can squeeze one more bit $b_1$ on its top level without causing interference. Therefore, we can achieve the symmetric rate of $\frac{3}{2}$ in this example, i.e., the 50\% improvement from the non-feedback rate of 1.

This scheme can be easily generalized to arbitrary $(n,m)$.
In the first stage, each transmitter sends $m$ bits on the upper
levels and $(n-m)$ bits on the lower levels. In the second stage,
each transmitter forwards the $m$ bits of the other user on the upper
levels and sends new $(n-m)$ private bits on the lower levels.
Then, each receiver can decode all of the $n$ bits sent in the
first stage and new $(n- m)$ private bits sent in the second
stage. Therefore, we can achieve:
\begin{align}
\label{eq:DICSymmWeak}
R_{\sf sym}  = \frac{n+(n-m)}{2} = n -
\frac{m}{2}.
\end{align}

\begin{remark}[\textbf{A Resource Hole Interpretation}]
Observe that all the resource levels are fully packed after applying the feedback scheme. Thus, feedback maximizes resource utilization to improve the performance significantly.
\end{remark}

\begin{remark}[\textbf{Exploiting Side Information}]
Another interpretation can be made to explain this gain. Recall that in the non-feedback case, the broadcast nature of the wireless medium precludes us from using the top level for one user when we are already using the bottom level for the other user. In contrast, if feedback is allowed, the top level can be used to improve the non-feedback rate. Suppose that transmitters 1 and 2 send $a_1$ and $b_1$ through their top levels, respectively. Receivers 1 and 2 then get the clean bits $a_1$ and $b_1$, respectively. With feedback, in the second stage, these bits $(a_1,b_1)$ can be exploited as \emph{side information} to refine the corrupted bits.
For example, with feedback transmitter 1 decodes the other user's bit $b_1$ and forwards it through the top level. This transmission allows receiver 1 to refine the corrupted bit $a_2$ from $a_2 \oplus b_1$ without causing interference to receiver
2, since receiver 2 already had the \emph{side information} of $b_1$ from the
previous broadcasting. We exploited the side information with the help of feedback to refine the corrupted bit without causing interference. The exploitation of side information was also observed and pointed out in network coding examples such as the butterfly network~\cite{ahlswede:it}, two-way relay channels~\cite{Wu:05}, general wireless networks~\cite{Katti:SIGCOMM06}, and broadcast erasure channels with feedback~\cite{Tassiulas:NetCom09}.
\end{remark}

\subsection{Optimality of the Achievable Scheme for the Deterministic IC}
\label{sec:Symm_OuterBound}

Now a natural question arises: is the scheme optimal? In this section, using the resource hole interpretation, we provide a positive conjecture on the optimality. Later in Section~\ref{sec:ElGamal-Costa}, we will provide a rigorous proof to settle this conjecture.




\textbf{From W to V Curve:} Fig.~\ref{fig:Symmetric_ResourceSharing} shows (i) the symmetric feedback rate~(\ref{eq:DICSymmStrong}), (\ref{eq:DICSymmWeak}) of the achievable scheme (representing the ``V'' curve); (ii) the non-feedback capacity~\cite{bresler:europe} (representing the ``W'' curve). Using the resource hole interpretation, we will provide intuition as to how we can go from the W curve to the V curve with feedback.

\begin{figure}[t]
\begin{center}
{\epsfig{figure=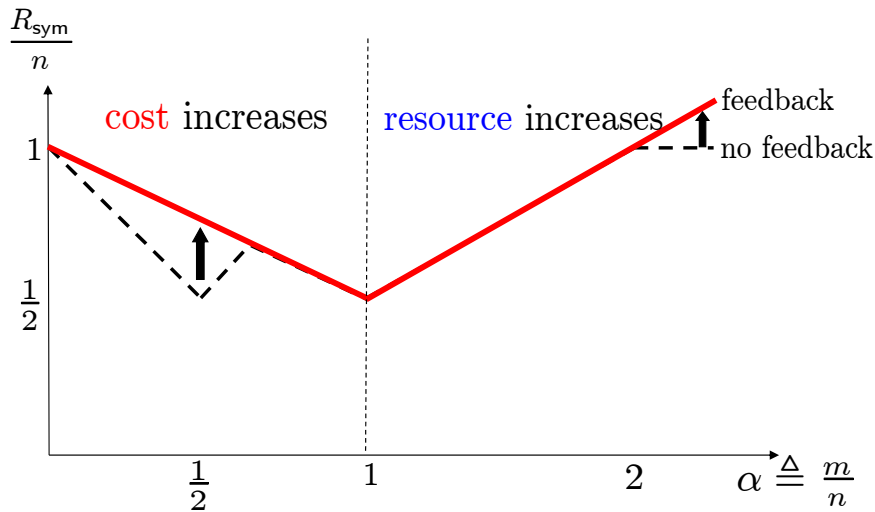, angle=0, width=0.7\textwidth}}
\end{center}
\caption{Symmetric feedback rate (\ref{eq:DICSymmStrong}), (\ref{eq:DICSymmWeak}) for the deterministic IC. Feedback maximizes resource utilization while it cannot reduce cost. The ``V'' curve is obtained when all of the resource levels are fully packed with feedback. This shows the optimality of the feedback scheme.} \label{fig:Symmetric_ResourceSharing}
\end{figure}

Observe that the total number of resource levels and transmission cost depend on $(n,m)$.
Specifically, suppose that the two senders employ the same transmission strategy to achieve the symmetric rate: using $x$ private and $y$ common levels. We then get:
\begin{align}
\begin{split}
&\textrm{\# of resource levels at each receiver}= \max (n,m ),\\
&\textrm{transmission cost}= 1\times x + 2 \times y.
\end{split}
\end{align}
Here notice that using a private level costs 1 level, while using a common level costs 2 levels. Now observe that as $\alpha=\frac{m}{n}$ grows: for $0 \leq \alpha \leq 1$, transmission cost increases; for $\alpha \geq 1$, the number of resource levels increases. Since all the resource levels are fully utilized with feedback, this observation implies that with feedback a total number of transmission bits must decrease when $0 \leq \alpha \leq 1$ (inversely proportional to transmission cost) and must increase when $\alpha \geq 1$ (proportional to the number of resource levels). This is reflected in the V curve. In contrast, in the non-feedback case, for some range of $\alpha$, resource levels are not fully utilized, as shown in the $\alpha=\frac{1}{2}$ example of Fig.~\ref{fig:Symmetric_Weak} $(a)$. This is reflected in the W curve.

\textbf{Why We Cannot Go Beyond the V Curve:}
While feedback maximizes resource utilization to fill up all of the resource holes, \emph{it cannot reduce transmission cost}. To see this, consider the example in Fig.~\ref{fig:Symmetric_Weak} $(b)$. Observe that even with feedback, a common bit still has to consume two levels at the two receivers. For example, the common bit $a_1$ needs to occupy the top level at receiver 1 in time 1; and the top level at receiver 2 in time 2. In time 1, while $a_1$ is received cleanly at receiver 1, it interferes with the private bit $b_2$. In order to refine $b_2$, receiver 2 needs to get $a_1$ cleanly and therefore needs to reserve one resource level for $a_1$. Thus, in order not to interfere with the private bit $b_1$, the common bit $a_1$ needs to consume a total of the two resource levels at the two receivers.
As mentioned earlier, assuming that transmission cost is not reduced, a total number of transmission bits is reflected in the V curve.
As a result, we cannot go beyond the ``V'' curve with feedback, showing the optimality of the achievable scheme. Later in Section~\ref{sec:ElGamal-Costa}, we will prove this rigorously.

\subsection{An Achievable Scheme for the Gaussian IC}

\begin{figure}[t]
\begin{center}
{\epsfig{figure=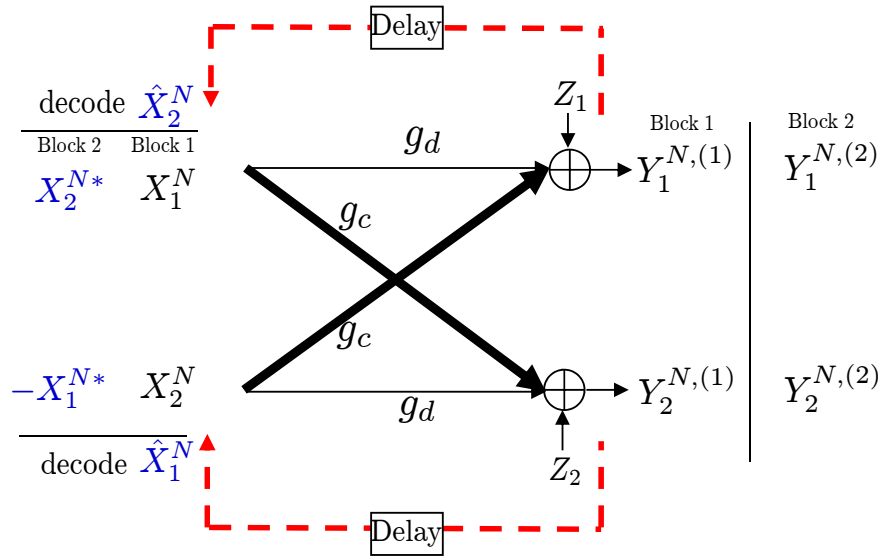, angle=0, width=0.7\textwidth}}
\end{center}
\caption{An Alamouti-based achievable scheme for the Gaussian IC: strong interference regime} \label{fig:Symmetric_Strong_Gaussian}
\end{figure}

Let us go back to the Gaussian channel. We will translate
the deterministic IC scheme to the Gaussian IC. Let us first consider the strong interference regime.

\textbf{Strong Interference Regime (${\sf INR} \geq  {\sf SNR}$):}
The structure of the transmitted signals in Fig.~\ref{fig:Symmetric_Strong_DIC} sheds some light on the Gaussian channel. Observe that in the second stage, each transmitter sends the other user's information sent in the first stage. This reminds us of \emph{Alamouti's scheme} \cite{Alamouti:jsac98}. The beauty of Alamouti's scheme is that received signals can be designed to be orthogonal during two time slots, although the signals in the first time slot are sent without any coding. This was exploited and pointed out in distributed space-time codes~\cite{Laneman:it03}. With Alamouti's scheme, transmitters are able to encode their messages so that received signals are orthogonal. Orthogonality between the two different signals guarantees complete removal of the interfering signal.

In accordance with the deterministic IC example, the scheme uses two stages (or blocks). In the first stage, transmitters 1 and 2 send codewords $X_{1}^{N}$ and $X_{2}^{N}$ with rates $R_1$ and $R_2$, respectively. In the second stage, using feedback, transmitters 1 and 2 decode $X_{2}^{N}$ and $X_{1}^{N}$, respectively.
This can be decoded if
\begin{align}
\label{eq:strongR_2c_constraint}
R_1, R_{2}  \leq \frac{1}{2} \log  \left( 1 + \mathsf{INR} \right) \;\;\textrm{bits/s/Hz}.
\end{align}
We are now ready to apply Alamouti's scheme. Transmitters 1 and 2 send $X_{2}^{N*}$ and $-X_{1}^{N*}$, respectively. Receiver 1 can then gather the two received signals: for $1 \leq i \leq N$,
\begin{align}
\left[
  \begin{array}{c}
    Y_{1i}^{(1)} \\
    Y_{1i}^{(2)*} \\
  \end{array}
\right]
= &\left[
    \begin{array}{cc}
      g_d & g_c \\
      -g_c^{*} & g_d^{*} \\
    \end{array}
  \right]
\left[
  \begin{array}{c}
    X_{1i} \\
    X_{2i} \\
  \end{array}
\right] +  \left[
  \begin{array}{c}
    Z_{1i}^{(1)} \\
    Z_{1i}^{(2)*} \\
  \end{array}
\right].
\end{align}
To extract $X_{1i}$, it multiplies the row vector orthogonal to the vector associated with $X_{2i}$ and therefore we get:
\begin{align}
\left[ \begin{array}{cc}
         g_d^{*} & -g_c
       \end{array}
 \right] \left[
  \begin{array}{c}
    Y_{1i}^{(1)} \\
    Y_{1i}^{(2)*} \\
  \end{array}
\right] = (|g_d|^2 + |g_c|^2) X_{1i} + g_d^* Z_{2i}^{(1)} - g_c Z_{1i}^{(2)*}.
\end{align}
The codeword $X_{1}^{N}$ can be decoded if
\begin{align}
\label{eq:strongR_1c_constraint}
R_{1} \leq \frac{1}{2} \log \left( 1 + \mathsf{SNR} +  \mathsf{INR} \right) \;\; \textrm{bits/s/Hz}.
\end{align}
Similar operations are done at receiver 2. Since  (\ref{eq:strongR_1c_constraint}) is implied by (\ref{eq:strongR_2c_constraint}), we get the desired result: the left term in (\ref{eq:SymmetricAchievableRate}).

\textbf{Weak Interference Regime (${\sf INR} \leq {\sf SNR}$):}
Unlike the strong interference regime, in the weak interference regime, there are two types of information: common and private information. A natural idea is to apply Alamouti's scheme only for common information and newly add private information. It was shown in~\cite{SuhTse:arix09} that this scheme can approximate the symmetric capacity to within $\approx 1.7$ bits/s/Hz. However, the scheme can be improved to reduce the gap further. Unlike the deterministic IC, in the Gaussian IC, private signals have some effects, i.e., these private signals cannot be completely ignored. Notice that the scheme includes \emph{decode-and-forward} operation at the transmitters after receiving the feedback. And so when each transmitter decodes the other user's common message while treating the other user's private signals as noise, private signals can incur performance loss.



This can be avoided by instead performing \emph{amplify-and-forward}: with feedback, the transmitters get the interference plus noise and then forward it subject to the power constraints. This transmission allows each receiver to refine its corrupted signal sent in the previous time, without causing significant interference.\footnote{In Appendix~\ref{Appendix:AchiSch_Symmetric}, we provide intuition behind this scheme.} Importantly, notice that this scheme does not require message-splitting. Even without splitting messages, we can refine the corrupted signals (see Appendix~\ref{Appendix:AchiSch_Symmetric} to understand this better). Therefore, there is no loss due to private signals.

Specifically, the scheme uses two stages. In the first stage, each transmitter $k$ sends codeword $X_k^{N}$ with rate $R_k$. In the second stage, with feedback transmitter 1 gets the interference plus noise:
\begin{align}
S_2^{N} = g_c X_2^{N} + Z_1^{(1),N}.
\end{align}
Now the complex conjugate technique based on Alamouti's scheme is applied to make $X_1^{N}$ and $S_2^{N}$ well separable.
Transmitters 1 and 2 send $\frac{S_2^{N*}}{\sqrt{1+ \sf INR}}$ and $-\frac{S_1^{N*}}{\sqrt{1+ \sf INR}}$, respectively, where $\sqrt{1+ \sf INR}$ is a normalization factor to meet the power constraint. Under Gaussian input distribution, we can compute the rate under MMSE demodulation: $\frac{1}{2} I(X_{1i}; Y_{1i}^{(1)}, Y_{2i}^{(2)})$.
Straightforward calculations give the desired result: the right term in (\ref{eq:SymmetricAchievableRate}). See Appendix~\ref{Appendix:AchiSch_Symmetric} for detailed computations.

\begin{remark}[\textbf{Amplify-and-Forward Reduces the Gap Further}]
\label{remark:AFbetter}
As mentioned earlier, unlike the decode-and-forward scheme, the amplify-and-forward scheme does not require message-splitting, thereby removing the effect of private signals. This improves the performance to reduce the gap further.
\end{remark}


\subsection{An Outer Bound}
Due to the overlap with the outer bound for the capacity region, we defer the proof to Theorem~\ref{theorem:outerbound} in Section~\ref{sec:OuterBoundRegion}.

\subsection{One-Bit Gap to the Symmetric Capacity}

Using the symmetric rate of~(\ref{eq:SymmetricAchievableRate}) and the outer bound of~(\ref{eq:SymmetricUpperBound}), we get:

\begin{align}
\begin{split}
2(\bar{C}_{\mathsf{sym}} - R_{\mathsf{sym}}) & \overset{(a)}{\leq}  \log \left(1 + \frac{\mathsf{SNR}}{ 1 + \mathsf{INR}} \right) + \log \left( 1 + \mathsf{SNR} + \mathsf{INR} + 2 \sqrt{\mathsf{SNR} \cdot \mathsf{INR} } \right) \\
& -\log \left(  \frac{(1+\mathsf{SNR}+\mathsf{INR})^2 - \frac{\mathsf{SNR}}{1+\mathsf{INR}} }{1+2\mathsf{INR}} \right) \\
&=  \log \left(  \frac{1+\mathsf{SNR}+\mathsf{INR}}{1+\mathsf{INR}} \cdot  \left( 1 + \mathsf{SNR} + \mathsf{INR} + 2 \sqrt{\mathsf{SNR} \cdot \mathsf{INR} } \right)  \right) \\
& + \log \left(  \frac{(1+2\mathsf{INR})(1+\mathsf{INR} ) }{(1+\mathsf{SNR}+\mathsf{INR})^2 (1+\mathsf{INR} ) - \mathsf{SNR}} \right) \\
&=  \log \left(  \frac{1+\mathsf{SNR}+\mathsf{INR} + 2 \sqrt{\mathsf{SNR} \cdot \mathsf{INR} } }{1+\mathsf{SNR}+ \mathsf{INR}} \cdot  \frac{1+2\mathsf{INR}}{1+\mathsf{INR} -\frac{\mathsf{SNR}}{(1+\mathsf{SNR}+\mathsf{INR})^2}  }   \right) \\
& \overset{(b)}{\leq} \log \left(  2 \cdot   \frac{ 2 \left(1+\mathsf{INR} - \frac{\mathsf{SNR}}{(1+\mathsf{SNR}+\mathsf{INR})^2} \right) -1 + \frac{2\mathsf{SNR}}{(1+\mathsf{SNR}+\mathsf{INR})^2} }{1+\mathsf{INR}-\frac{\mathsf{SNR}}{(1+\mathsf{SNR}+\mathsf{INR})^2}}    \right) \\
& =  \log \left(  2 \cdot  \left \{  2 -  \left(  \frac{1-\frac{2\mathsf{SNR}}{(1+\mathsf{SNR}+\mathsf{INR})^2} }{1+\mathsf{INR}-\frac{\mathsf{SNR}}{(1+\mathsf{SNR}+\mathsf{INR})^2}}  \right) \right\} \right) \\
& \overset{(c)}{\leq}  \log 4 = 2,
\end{split}
\end{align}
where $(a)$ follows from choosing trivial maximum and minimum values of the outer bound~(\ref{eq:SymmetricUpperBound}) and the lower bound~(\ref{eq:SymmetricAchievableRate}), respectively; $(b)$ follows from $1 + {\sf SNR} + {\sf INR} + 2 \sqrt{ \sf SNR \cdot \sf INR} \leq 2 (1 + {\sf SNR} + {\sf INR})$; and $(c)$ follows from $(1 + {\sf SNR} + {\sf INR})^2 \geq 2 {\sf SNR}$ and $\frac{\sf SNR}{ (1 + {\sf SNR} + {\sf INR} )^2 } \leq 1$.

Fig. \ref{fig:Gap_Proposed} shows a numerical result for the gap between the inner and outer bounds. Notice that the gap is upper-bounded by exactly one bit. The worst-case gap occurs when ${\sf SNR } \approx {\sf INR}$ and these values go to infinity. Also note that in the strong interference regime, the gap approaches 0 with the increase of $\sf SNR$ and $\sf INR$, while in the weak interference regime, the gap does not vanish. For example, when $\alpha = \frac{1}{2}$, the gap is around 0.5 bits.

\begin{remark}[\textbf{Why does a 1-bit gap occur?}] Observe in Figs.~\ref{fig:Symmetric_Strong_Gaussian} and~\ref{fig:Symmetric_Unified} that  the transmitted signals of the two senders are \emph{uncorrelated} in our scheme. The scheme completely loses power gain (also called beamforming gain). On the other hand, when deriving the outer bound of~(\ref{eq:SymmetricUpperBound}), we allow for arbitrary correlation between the transmitters. Thus, the 1-bit gap is based on \emph{the outer bound}. In the actual system, correlation is in-between and therefore one can expect that an actual gap to the capacity is less than 1 bit.

Beamforming gain is important only when $\sf SNR$ and $\sf INR$ are quite close, i.e., $\alpha \approx 1$. This is because when $\alpha =1$, the interference channel is equivalent to the multiple access channel where Ozarow's scheme~\cite{Ozarow:it} and Kramer's scheme~\cite{Kramer:it02} (that capture beamforming gain) are optimal. In fact, the capacity theorem in~\cite{Kramer:it04} has shown that Kramer's scheme is optimal for one specific case of ${\sf INR} = {\sf SNR} - \sqrt{2 {\sf  SNR} }$, although it is arbitrarily far from optimality for the other cases. This observation implies that our proposed scheme can be improved further.
\end{remark}

\section{Capacity Region to Within 2 Bits}
\label{sec:FB_CapacityRegion}

\subsection{An Achievable Rate Region}
We have developed an achievable scheme meant for the symmetric rate and provided a resource hole interpretation. For the case of the capacity region, we find that while this interpretation can also be useful, the two-staged scheme cannot be applied. A new achievable scheme needs to be developed for the region characterization.

To see this, let us consider a deterministic IC example in Fig.~\ref{fig_AchievabilityIdea} where an infinite number of stages need to be employed to achieve a corner point of $(2,1)$ with feedback.
\begin{figure}[!htp]
\begin{center}
{\epsfig{figure=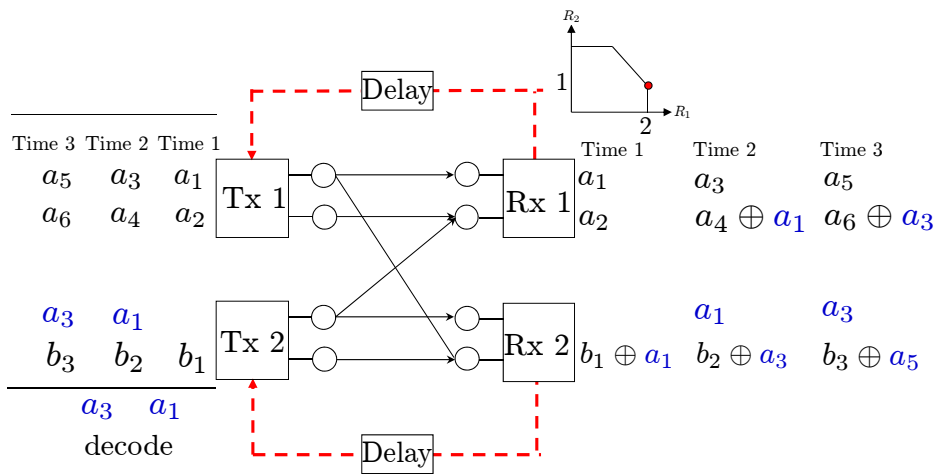, angle=0, width=0.75\textwidth}}
\end{center}
\caption{A deterministic IC example where an infinite number of stages need to be employed to achieve the rate pair of $(2,1)$ with feedback.} \label{fig_AchievabilityIdea}
\end{figure}
Observe that to guarantee $R_1 = 2$, transmitter 1 needs to send 2 bits every time slot. Once transmitter 1 sends $(a_1,a_2)$, transmitter 2 cannot use its top level since the transmission causes interference to receiver 1. It can use only the bottom level to send information.
This transmission however suffers from interference: receiver 2 gets the interfered signal $b_1 \oplus a_1$. We will show that this corrupted bit can be refined with feedback.
In time 2, transmitter 2 can decode $a_1$ with feedback. In an effort to achieve the rate pair of $(2,1)$, transmitter 1 sends $(a_3,a_4)$ and transmitter 2 sends $b_2$ on the bottom level. Now apply the same idea used in the symmetric case: transmitter 2 sends the other user's information $a_1$ on the top level. This transmission allows receiver 2 to refine the corrupted signal $b_1$ without causing interference to receiver 1, since receiver 1 already had $a_1$ as \emph{side information}. Notice that during the two time slots, receiver 1 can decode 4 bits (2 bits/time), while receiver 2 can decode 1 bits (0.5 bits/time). The point $(2,1)$ is not achieved yet due to unavoidable loss occurred in time 1. This loss, however, can be amortized by iterating the same operation.
As this example shows, the previous two-staged scheme needs to be modified so as to incorporate an infinite number of stages.

Let us apply this idea to the Gaussian channel.
The use of an infinite number of stages motivates the need for employing \emph{block Markov encoding} \cite{Cover:it79, Cover:it81}. Similar to the symmetric case, we can now think of two possible schemes: (1) decode-and-forward (with message-splitting); and (2) amplify-and-forward (without message-splitting).
As pointed out in Remark~\ref{remark:AFbetter}, in the Gaussian channel, private signals cannot be completely ignored, thereby incurring performance loss, thus the amplify-and-forward scheme without message-splitting has better performance.
However, it requires heavy computations to compute the rate region, so we focus on the decode-and-forward scheme, although it induces a larger gap.
As for a decoding operation, we employ backward decoding \cite{Kuhlmann:it89, Tuninetti:isit07}.

Here is the outline of our scheme. We employ block Markov encoding with a total size $B$ of blocks. In block 1, each transmitter splits its own message into common and private parts and then sends a codeword superimposing the common and private messages. For power splitting, we adapt the idea of the simplified Han-Kobayashi scheme~\cite{dtse:it07} where private power is set such that a private signal is seen below the noise level at the other receiver. In block 2, with feedback, each transmitter decodes the other user's common message (sent in block 1) while treating the other user's private signal as noise. Two common messages are then available at the transmitter: (1) its own message; and (2) the other user's message decoded with the help of feedback. Conditioned on these two common messages, each transmitter generates new common and private messages. It then sends the corresponding codeword. Each transmitter repeats this procedure until block $B-1$. In the last block $B$, to facilitate backward decoding, each transmitter sends the predetermined common message and a new private message. Each receiver waits until total $B$ blocks have been received and then performs backward decoding.
We will show that this scheme enables us to obtain an achievable rate region that approximates the capacity region.

\begin{theorem}
\label{theorem:achievableregion}
The feedback capacity region includes the set $\mathcal{R}$ of $(R_1, R_2)$ such that
\begin{align}
\label{eq:achieve_R1_1}    R_1 &\leq \log \left( 1+ \mathsf{SNR}_1  + \mathsf{INR}_{21} + 2 \rho \sqrt{ \mathsf{SNR}_1  \cdot \mathsf{INR}_{21}}
 \right) - 1 \\
 \label{eq:achieve_R1_2}
R_1 & \leq \log \left( 1 +  (1- \rho)  \mathsf{INR}_{12} \right)
 + \log \left( 2 +  \frac{ \mathsf{SNR}_1}{ \mathsf{INR}_{12}} \right) - 2
 \\
 \label{eq:achieve_R2_1}
    R_2 &\leq \log \left( 1+ \mathsf{SNR}_2  + \mathsf{INR}_{12} + 2 \rho \sqrt{ \mathsf{SNR}_2  \cdot \mathsf{INR}_{12}}
 \right) - 1 \\
 \label{eq:achieve_R2_2}
R_2 & \leq \log \left( 1 +  (1- \rho)  \mathsf{INR}_{21} \right)
 + \log \left( 2 +  \frac{ \mathsf{SNR}_2}{ \mathsf{INR}_{21}} \right) - 2
 \\
 \label{eq:achieve_R12_1}
R_1 + R_2 & \leq \log \left( 2 +  \frac{ \mathsf{SNR}_1}{ \mathsf{INR}_{12}} \right)  +  \log \left( 1+ \mathsf{SNR}_2  + \mathsf{INR}_{12} + 2 \rho \sqrt{ \mathsf{SNR}_2  \cdot \mathsf{INR}_{12}}
 \right) - 2 \\
 \label{eq:achieve_R12_2}
R_1 + R_2 & \leq \log \left( 2 +  \frac{ \mathsf{SNR}_2}{ \mathsf{INR}_{21}} \right)  +  \log \left( 1+ \mathsf{SNR}_1  + \mathsf{INR}_{21} + 2 \rho \sqrt{ \mathsf{SNR}_1  \cdot \mathsf{INR}_{21}}
 \right) - 2
\end{align}
for $0 \leq \rho \leq 1$.
\end{theorem}
\begin{proof}
Our achievable scheme is generic, not limited to the Gaussian IC. We therefore characterize an achievable rate region for discrete memoryless ICs and then choose an appropriate joint distribution to obtain the desired result. In fact, this generic scheme can also be applied to El Gamal-Costa deterministic IC (to be described in Section~\ref{sec:ElGamal-Costa}).

\begin{lemma}
\label{lemma:feedbackachievable}
The feedback capacity region of the two-user discrete memoryless IC includes the set of  $(R_1,R_2)$ such that
\begin{align}
  R_1 & \leq  I(U,U_2,X_1;Y_1) \\
  R_1 & \leq I(U_1;Y_2|U,X_2) +  I(X_1;Y_1|U_1,U_2,U)  \\
    R_2 &\leq I(U,U_1,X_2;Y_2) \\
  R_2 &\leq I(U_2;Y_1|U,X_1) +  I(X_2;Y_2|U_1,U_2,U)  \\
  R_1 + R_2 &\leq  I(X_1;Y_1|U_1,U_2,U) + I(U,U_1,X_2;Y_2) \\
  R_1 + R_2 & \leq I(X_2;Y_2|U_1,U_2,U) + I(U,U_2,X_1;Y_1),
\end{align}
over all joint distributions $p(u) p(u_1|u) p(u_2|u) p(x_1 |u_1, u) p(x_2 |u_2, u)$.
\end{lemma}
\begin{proof}
See Appendix~\ref{Appendix:lemmaachievable}.
\end{proof}

Now we will choose the following Gaussian input distribution to complete the proof: $\forall k=1,2,$
\begin{align}
U \sim \mathcal{CN}(0,\rho); U_k \sim \mathcal{CN}(0,\lambda_{ck}); X_{pk} \sim \mathcal{CN}(0,\lambda_{pk}),
\end{align}
where $X_k = U + U_k + X_{kp}$; $\lambda_{ck}$ and $\lambda_{pk}$ indicate the powers allocated to the common and private message of transmitter $k$, respectively; and $(U, U_k, X_{kp})$'s are independent.
By symmetry, it suffices to prove~(\ref{eq:achieve_R1_1}), (\ref{eq:achieve_R1_2}) and (\ref{eq:achieve_R12_1}).

To prove~(\ref{eq:achieve_R1_1}), consider $ I(U,U_2,X_1;Y_1) = h(Y_1) - h(Y_1|U,U_2, X_1)$. Note
\begin{align}
|K_{Y_1 |X_1, U_2, U}| = 1+  \lambda_{p2} \mathsf{INR}_{21}.
\end{align}
As mentioned earlier, for power splitting, we adapt the idea of the simplified Han-Kobayashi scheme~\cite{dtse:it07}.
We set private power such that the private signal appears below the noise level at the other receiver. This idea mimics that of the deterministic IC example where the private bit is below the noise level so that it is invisible. The remaining power is assigned to the common message. Specifically, we set:
\begin{align}
\begin{split}
\label{eq:weakpowersplit}
\lambda_{p2} = \min \left( \frac{1}{\mathsf{INR}_{21}}, 1 \right), \;\; \lambda_{c2} = 1- \lambda_{p2},
\end{split}
\end{align}
This choice gives
\begin{align}
\label{eq:mutualY1full}
    I(U,U_2,X_1;Y_1) = \log \left( 1+ \mathsf{SNR}_1  + \mathsf{INR}_{21} + 2 \rho \sqrt{ \mathsf{SNR}_1  \cdot \mathsf{INR}_{21}}
 \right) - 1,
\end{align}
which proves $(\ref{eq:achieve_R1_1})$. With the same power setting, we can compute:
\begin{align}
 & I(U_1;Y_2|U, X_2) = \log \left( 1 +  (1-\rho) \mathsf{INR}_{12} \right) - 1, \\
 \label{eq:mutualInfoprivate}
 &I(X_1;Y_1|U,U_1, U_2) = \log \left( 2 +  \frac{ \mathsf{SNR}_1}{ \mathsf{INR}_{12}} \right) - 1.
\end{align}
This proves $(\ref{eq:achieve_R1_2})$. Lastly, by (\ref{eq:mutualY1full}) and (\ref{eq:mutualInfoprivate}), we prove $(\ref{eq:achieve_R12_1})$.
\end{proof}

\begin{remark}[\textbf{Three Types of Inequalities}] In the non-feedback case, it is shown in~\cite{dtse:it07} that an approximate capacity region is characterized by five types of inequalities including the bounds for $2R_1 + R_2$ and $R_1 + 2R_2$. In contrast, in the feedback case, our achievable rate region is described by only three types of inequalities.\footnote{It is still unknown whether or not the exact feedback capacity region includes only three types of inequalities.} In Section \ref{sec:2R1R2boundmissing}, we will provide qualitative insights as to why the $2R_1 + R_2$ bound is missing with feedback.
\end{remark}

\begin{remark}[Connection to Related Work~\cite{Tuninetti:isit07}] Our achievable scheme is essentially the same as the scheme introduced by Tuninetti~\cite{Tuninetti:isit07} in a sense that the three techniques (message-splitting, block Markov encoding and backward decoding) are jointly employed.\footnote{The author in~\cite{Tuninetti:isit07} considers a different context: the conferencing encoder problem. However, Prabhakaran and Viswanath~\cite{Vinod:arix09} have made an interesting connection between the feedback problem and the conferencing encoder problem. See~\cite{Vinod:arix09} for details.}
However, the scheme in~\cite{Tuninetti:isit07} uses five auxiliary random variables requiring further optimization. On the other hand, we obtain an explicit rate region by reducing those five auxiliary random variables into three and then choosing a joint input distribution appropriately.
\end{remark}

\subsection{An Outer Bound Region}
\label{sec:OuterBoundRegion}

\begin{theorem}
\label{theorem:outerbound}
The feedback capacity region is included by the set $\overline{\mathcal{C}}$ of $(R_1, R_2)$ such that
\begin{align}
\label{eq:outerR1_1}
R_1 & \leq  \log \left( 1+ \mathsf{SNR}_1  + \mathsf{INR}_{21} + 2 \rho \sqrt{ \mathsf{SNR}_1  \cdot \mathsf{INR}_{21}}
 \right) \\
\label{eq:outerR1_2}
R_1 & \leq \log \left( 1 +  (1- \rho^2) \mathsf{INR}_{12} \right)
 + \log \left( 1 +  \frac{ (1- \rho^2) \mathsf{SNR}_1}{ 1 + (1- \rho^2) \mathsf{INR}_{12}} \right)
 \\
 \label{eq:outerR2_1}
R_2 & \leq  \log \left( 1+ \mathsf{SNR}_2  + \mathsf{INR}_{12} + 2 \rho \sqrt{ \mathsf{SNR}_2  \cdot \mathsf{INR}_{12}}
 \right) \\
 \label{eq:outerR2_2}
R_2 & \leq \log \left( 1 +  (1- \rho^2) \mathsf{INR}_{21} \right)
 + \log \left( 1 +  \frac{ (1- \rho^2) \mathsf{SNR}_2}{ 1 + (1- \rho^2) \mathsf{INR}_{21} } \right)
 \\
 \label{eq:outerR1_R2_1}
R_1 + R_2 & \leq \log \left( 1 +  \frac{ (1- \rho^2) \mathsf{SNR}_1}{ 1 + (1- \rho^2) \mathsf{INR}_{12} } \right) +  \log \left( 1+ \mathsf{SNR}_2  + \mathsf{INR}_{12} + 2 \rho \sqrt{ \mathsf{SNR}_2  \cdot \mathsf{INR}_{12} }
 \right) \\
 \label{eq:outerR1_R2_2}
R_ 1 + R_2 & \leq \log \left( 1 +  \frac{ (1- \rho^2) \mathsf{SNR}_2}{ 1 + (1- \rho^2) \mathsf{INR}_{21} } \right) +  \log \left( 1+ \mathsf{SNR}_1  + \mathsf{INR}_{21} + 2 \rho \sqrt{ \mathsf{SNR}_1  \cdot \mathsf{INR}_{21} }
 \right)
\end{align}
for $0 \leq \rho \leq 1$.
\end{theorem}
\begin{proof}
By symmetry, it suffices to prove the bounds of~(\ref{eq:outerR1_1}), (\ref{eq:outerR1_2}) and (\ref{eq:outerR1_R2_1}). The bounds of (\ref{eq:outerR1_1}) and (\ref{eq:outerR1_2}) are nothing but cutset bounds. Hence, proving the non-cutset bound of (\ref{eq:outerR1_R2_1}) is the main focus of this proof. Also recall that this non-cutset bound is used to obtain the outer bound of~(\ref{eq:SymmetricUpperBound}) for the symmetric capacity in Theorem~\ref{theorem-symmetric}. We go through the proof of~(\ref{eq:outerR1_1}) and~(\ref{eq:outerR1_2}). We then focus on the proof of~(\ref{eq:outerR1_R2_1}), where we will also provide insights as to the proof idea.

\textbf{Proof of (\ref{eq:outerR1_1}):} Starting with Fano's inequality, we get:
\begin{align*}
\begin{split}
N( R_1 - \epsilon_N) \leq I(W_1;Y_1^{N}) \overset{(a)}{\leq} \sum  [h(Y_{1i}) - h(Z_{1i})],\\
\end{split}
\end{align*}
where $(a)$ follows from the fact that conditioning reduces entropy.
Assume that $X_1$ and $X_2$ have covariance $\rho$, i.e., $E[X_1X_2^*]=\rho$. Then, we get:
\begin{align}
\begin{split}
\label{eq:entropyofY1}
h(Y_1) 
\leq \log 2 \pi e \left( 1 + \mathsf{SNR}_1 + \mathsf{INR}_{21} + 2 |\rho| \sqrt{\mathsf{SNR}_1 \cdot \mathsf{INR}_{21} } \right).
\end{split}
\end{align}
If $(R_1,R_2)$ is achievable, then $\epsilon_N \rightarrow 0$ as $N \rightarrow \infty$. Therefore, we get the desired bound:
\begin{align}
R_1 \leq h(Y_1) - h(Z_1) \leq \log \left( 1+ \mathsf{SNR}_1  + \mathsf{INR}_{21} + 2 |\rho| \sqrt{ \mathsf{SNR}_1  \cdot \mathsf{INR}_{21} }
 \right).
\end{align}

\textbf{Proof of (\ref{eq:outerR1_2}):} Starting with Fano's inequality, we get:
\begin{align*}
\begin{split}
N&(R_1 - \epsilon_N) \leq I(W_1;Y_1^N, Y_2^N, W_2) \\
&\overset{(a)}= \sum [h(Y_{1i},Y_{2i}|W_2,Y_1^{i-1},Y_2^{i-1}) - h(Z_{1i}) - h (Z_{2i}) ]  \\
&\overset{(b)}{=} \sum [h(Y_{1i},Y_{2i}|W_2, Y_1^{i-1},Y_2^{i-1},X_{2}^{i}) - h(Z_{1i}) - h (Z_{2i}) ] \\
&\overset{(c)}{=} \sum [h(Y_{2i}|W_2,Y_1^{i-1},Y_2^{i-1}, X_{2}^{i}) - h (Z_{2i})]  + \sum [h(Y_{1i}|W_2, Y_1^{i-1},Y_2^{i-1}, X_{2}^{i},Y_{2i},S_{1}^{i}) - h(Z_{1i}) ]  \\
&\overset{(d)}{\leq} \sum \left[ h(Y_{2i}|X_{2i}) - h (Z_{2i}) + h(Y_{1i}|X_{2i},S_{1i}) - h (Z_{1i}) \right]
\end{split}
\end{align*}
where ($a$) follows from the fact that $W_1$ is independent from $W_2$ and $h(Y_1^{N},Y_2^{N}|W_1,W_2)= h(Y_1^N,S_1^N|W_1,W_2) = \sum [h(Z_{1i}) + h(Z_{2i})]$ (see Claim \ref{claim-5}); $(b)$ follows from the fact that $X_{2}^{i}$ is a function of $(W_2,Y_2^{i-1})$; ($c$) follows from the fact that $S_{1}^{i}$ is a function of $(Y_{2}^{i},X_{2}^{i})$; ($d$) follows from the fact that conditioning reduces entropy. Hence, we get the desired result:
\begin{align*}
\begin{split}
R_1 &\leq  h(Y_{2}|X_{2}) - h (Z_{2}) + h(Y_{1}|X_{2},S_{1}) - h (Z_{1}) \\
&\overset{(a)}{\leq} \log \left( 1 +   (1- |\rho|^2) \mathsf{INR}_{12} \right) +   \log \left( 1 +  \frac{ (1- |\rho|^2) \mathsf{SNR}_1}{ 1 + (1- |\rho|^2) \mathsf{INR}_{12}} \right)
\end{split}
\end{align*}
where $(a)$ follows from the fact that
\begin{align}
h(Y_{2}|X_{2}) &\leq \log 2 \pi e \left( 1 +   (1- |\rho|^2) \mathsf{INR}_{12} \right), \\
\label{eq:h_Y1_X2S1}
h(Y_1|X_2,S_1) & \leq \log 2 \pi e \left( 1+  \frac{(1-|\rho|^2)\mathsf{SNR}_1}{1 + (1-|\rho|^2) \mathsf{INR}_{12} } \right).
\end{align}
The inequality of~(\ref{eq:h_Y1_X2S1}) is obtained as follows. Given $(X_2,S_1)$, the variance of $Y_1$ is upper-bounded by
\begin{align*}
\begin{split}
\textrm{Var} \left[ Y_1| X_2, S_1 \right] &\leq K_{Y_1} - K_{Y_1 (X_2, S_1)}K_{(X_2,S_1)}^{-1}K_{Y_1 (X_2, S_1)}^{*},
\end{split}
\end{align*}
where
\begin{align}
\begin{split}
K_{Y_1} &= E\left[ |Y_1|^2 \right] = 1  + \mathsf{SNR}_1 + \mathsf{INR}_{21} +  \rho g_{11}^* g_{21} + \rho^* g_{11}g_{21}^*, \\
K_{Y_1 (X_2,S_1)} &= E \left[ Y_1 [X_2^*, S_1^*]  \right] = \left[\rho g_{11} + g_{21}, g_{12}^* g_{11} + \rho^* g_{21} g_{12}^* \right], \\
K_{(X_2,S_1)} & = E \left[ \left[
                             \begin{array}{cc}
                               |X_2|^2 & X_2 S_1^* \\
                               X_2^* S_1 & |S_1|^2 \\
                             \end{array}
                           \right]
 \right] = \left[
             \begin{array}{cc}
               1 & \rho^* g_{12}^* \\
               \rho g_{12} & 1+ \mathsf{INR}_{12} \\
             \end{array}
           \right].
\end{split}
\end{align}
By further calculation, we can get~(\ref{eq:h_Y1_X2S1}).

\textbf{Proof of (\ref{eq:outerR1_R2_1}):} The proof idea is based on the genie-aided argument~\cite{ElGamal:it82}. However, finding an appropriate genie is not simple since there are many possible combinations of the random variables. The deterministic IC example in Fig.~\ref{fig:Symmetric_Weak} $(b)$ gives insights into this. Note that providing $a_1$ and $(b_1,b_2,b_3)$ to receiver 1 does not increase the rate $R_1$, i.e., these are \emph{useless gifts}. This may motivate us to choose a genie as $(g_{12} X_1, W_2)$. However, in the Gaussian channel, providing $g_{12} X_1$ is equivalent to providing $X_1$. This is of course too much information, inducing a loose upper bound. Inspired by the technique in~\cite{dtse:it07}, we instead consider a noisy version of $g_{12} X_1$:  \begin{align}
S_1 =  g_{12} X_1 + Z_2.
\end{align}
Intuition behind this is that we cut off $g_{12} X_1$ at the noise level. Indeed this matches intuition in the deterministic IC. This genie together with $W_2$ turns out to lead to the desired tight upper bound.

On top of the genie-aided argument, we need more techniques. In the feedback problem, the functional relationship between the random variables is more complicated and thus needs be well explored. We identify several functional relationships through Claims~\ref{claim-5}, \ref{claim-4} and~\ref{claim-3_Gaussian}, which turn out to play a significant role in proving the non-cutset bound of~(\ref{eq:outerR1_R2_1}).

Starting with Fano's inequality, we get:
\begin{align*}
\begin{split}
N&(R_1 + R_2 - \epsilon_N) \leq I(W_1;Y_1^{N}) + I(W_2;Y_2^{N})  \\
&\overset{(a)}{\leq} I(W_1;Y_1^{N},S_1^{N}, W_2) + I(W_2;Y_2^{N}) \\
&\overset{(b)}= h(Y_1^{N},S_1^{N}|W_2)  - h(Y_1^{N},S_1^{N}|W_1,W_2) + I(W_2;Y_2^{N}) \\
&\overset{(c)}{=} h(Y_1^{N},S_1^{N}|W_2) - \sum \left[ h(Z_{1i}) + h(Z_{2i}) \right]  + I(W_2;Y_2^{N}) \\
&\overset{(d)}{=} h(Y_1^{N}|S_1^{N},W_2) - \sum  h(Z_{1i}) + h(Y_2^{N})  - \sum h(Z_{2i})  \\
&\overset{(e)}{=} h(Y_1^{N}|S_1^{N},W_2,X_2^{N}) - \sum h(Z_{1i}) + h(Y_2^{N}) - \sum h(Z_{2i})  \\
&\overset{(f)}{\leq} \sum_{i=1}^{N} \left[ h(Y_{1i}|S_{1i},X_{2i}) - h(Z_{1i}) + h(Y_{2i}) - h(Z_{2i}) \right]
\end{split}
\end{align*}
where ($a$) follows from the fact that adding information increases mutual information (providing a \emph{genie}); ($b$) follows from the independence of $W_1$ and $W_2$; ($c$) follows from $h(Y_1^{N},S_1^{N}|W_1,W_2) =\sum \left[ h(Z_{1i}) + h(Z_{2i}) \right]$ (see Claim \ref{claim-5}); $(d)$ follows from  $h(S_1^N|W_2)=h(Y_2^N|W_2)$ (see Claim \ref{claim-4}); ($e$) follows from the fact that $X_2^{N}$ is a function of $(W_2,S_1^{N-1})$ (see Claim \ref{claim-3_Gaussian}); ($f$) follows from the fact that conditioning reduces entropy.

Hence, we get
\begin{align*}
R_1 + R_2 &\leq  h(Y_{1}|S_{1},X_{2}) - h(Z_1) + h(Y_{2}) - h(Z_2).
\end{align*}
Note that
\begin{align}
\begin{split}
\label{eq:entropyofY2}
h(Y_2) 
\leq \log 2 \pi e \left( 1 + \mathsf{SNR}_2 + \mathsf{INR}_{12} + 2 |\rho| \sqrt{\mathsf{SNR}_2 \cdot \mathsf{INR}_{12} } \right).
\end{split}
\end{align}
From (\ref{eq:h_Y1_X2S1}) and (\ref{eq:entropyofY2}),  we get the desired upper bound.

\begin{claim}
\label{claim-5}
$h(Y_1^{N},S_1^{N}|W_1,W_2) = \sum \left[ h(Z_{1i}) + h(Z_{2i}) \right].$
\end{claim}
\begin{proof}
\begin{align*}
\begin{split}
h&(Y_1^{N},S_1^{N}|W_1,W_2) = \sum h(Y_{1i},S_{1i}|W_1,W_2,Y_1^{i-1},S_1^{i-1}) \\
&\overset{(a)}{=} \sum h(Y_{1i},S_{1i}|W_1,W_2,Y_1^{i-1},S_1^{i-1},X_{1i},X_{2i}) \\
&\overset{(b)}{=}  \sum h(Z_{1i},Z_{2i}|W_1,W_2,Y_1^{i-1},S_1^{i-1},X_{1i},X_{2i}) \\
&\overset{(c)}{=}  \sum \left[ h(Z_{1i}) + h(Z_{2i}) \right],
\end{split}
\end{align*}
where ($a$) follows from the fact that $X_{1i}$ is a function of $(W_1, Y_1^{i-1})$ and $X_{2i}$ is a function of $(W_2,S_{1}^{i-1})$ (by Claim \ref{claim-3_Gaussian});
($b$) follows from the fact that $Y_{1i} = g_{11} X_{1i} + g_{21} X_{2i} + Z_{1i}$ and $S_{1i} = g_{12} X_{1i} + Z_{2i}$; ($c$) follows from the memoryless property of the channel and the independence assumption of $Z_{1i}$ and $Z_{2i}$.
\end{proof}

\begin{claim}
\label{claim-4}
$h(S_1^{N}|W_2) = h(Y_2^{N}|W_2).$
\end{claim}
\begin{proof}
\begin{align*}
\begin{split}
h&(Y_2^{N}|W_2) = \sum h(Y_{2i}|Y_2^{i-1},W_2) \\
&\overset{(a)}{=} \sum h(S_{1i}|Y_2^{i-1},W_2) \\
&\overset{(b)}{=}  \sum h(S_{1i}|Y_2^{i-1},W_2,X_2^{i},S_1^{i-1}) \\
&\overset{(c)}{=}  \sum h(S_{1i}|W_2,S_1^{i-1}) =h(S_1^{N}|W_2),
\end{split}
\end{align*}
where ($a$) follows from the fact that $Y_{2i}$ is a function of $(X_{2i},S_{1i})$ and $X_{2i}$ is a function of $(W_2, Y_2^{i-1})$;
($b$) follows from the fact that $X_2^{i}$ is a function of $(W_2,Y_2^{i-1})$ and $S_1^{i-1}$ is a function of $(Y_2^{i-1},X_2^{i-1})$; ($c$) follows from the fact that $Y_{2}^{i-1}$ is a function of $(X_2^{i-1},S_1^{i-1})$ and $X_2^{i}$ is a function of  $(W_2,S_1^{i-1})$ (by Claim \ref{claim-3_Gaussian}).
\end{proof}

\begin{claim}
\label{claim-3_Gaussian}
For all $i\geq 1$, $X_1^{i}$ is a function of $(W_1,S_2^{i-1})$ and $X_2^{i}$ is a function of $(W_2,S_1^{i-1})$.
\end{claim}
\begin{proof}
By symmetry, it is enough to prove only one. Notice that $X_2^{i}$ is a function of $(W_2, Y_2^{i-1})$ and $Y_{2}^{i-1}$ is a function of $(X_2^{i-1},S_1^{i-1}$). Hence, $X_2^{i}$ is a function of $(W_2, X_2^{i-1}, S_1^{i-1})$. Iterating the same argument, we conclude that $X_{2}^{i}$ is a function of $(W_2, X_{21}, S_1^{i-1})$. Since $X_{21}$ depends only on $W_2$, we complete the proof.
\end{proof}
\end{proof}

\subsection{2-Bit Gap to the Capacity Region}

\begin{theorem}
\label{theorem-2bitgap}
The gap between the inner and upper bound regions (given in Theorems~\ref{theorem:achievableregion} and~\ref{theorem:outerbound}) is at most $2$ bits/s/Hz/user:
\begin{align}
\mathcal{R} \subseteq \mathcal{C} \subseteq \overline{\mathcal{C}} \subseteq  \mathcal{R} \oplus \left( [0,2] \times [0,2] \right).
\end{align}
\end{theorem}
\begin{proof}
The proof is immediate by Theorem~\ref{theorem:achievableregion} and~$\ref{theorem:outerbound}$. We define $\delta_1$ to be the difference between $\min \left\{ (\ref{eq:outerR1_1}), (\ref{eq:outerR1_2}) \right\}$ and $\min \left\{ (\ref{eq:achieve_R1_1}), (\ref{eq:achieve_R1_2}) \right\}$. Similarly, we define $\delta_2$ and $\delta_{12}$. Straightforward computation gives
\begin{align*}
\begin{split}
\delta_1 &\leq \max \left \{  1, \log \left( 1 +   \frac{\mathsf{SNR}_1}{ 1 + \mathsf{INR}_{12} } \right) - \log \left( 2 +  \frac{ \mathsf{SNR}_1}{ \mathsf{INR}_{12} } \right) + 2 \right \} \leq 2.
\end{split}
\end{align*}
Similarly, we get $\delta_2 \leq 2$ and $\delta_{12} \leq 2$. This completes the proof.
\end{proof}

\begin{remark}[\textbf{Why does a 2-bit gap occur?}]
\label{remark:Why2Bits}
The achievable scheme meant for the capacity region involves message-splitting. As mentioned in Remark~\ref{remark:AFbetter}, message-splitting incurs some loss in the process of decoding the common message while treating private signals as noise.
Accounting for the effect of private signals, the effective noise power becomes double, thus incurring a 1-bit gap. The other 1-bit gap comes from a relay structure of the feedback IC. To see this, consider an extreme case where user 2's rate is completely ignored. In this case, we can view the $[transmitter 2, receiver 2]$ communication pair as a single relay which only helps the $[transmitter 1, receiver 1]$ communication pair. It has been shown in~\cite{Salman:allterton07} that for this single relay Gaussian channel, the worst-case gap between the best known inner bound~\cite{Cover:it79} and the outer bound is 1 bit/s/Hz. This incurs the other 1-bit gap. This 2-bit gap is based on the outer bound region in Theorem~\ref{theorem:outerbound}, which allows for arbitrary correlation between the transmitters. So, one can expect that an actual gap to the capacity region is less than 2 bits.
\end{remark}

\begin{remark}[\textbf{Reducing the gap}]
As discussed, the amplify-and-forward scheme has the potential to reduce the gap. However, due to the inherent relay structure, reducing the gap into a less-than-one bit is challenging. As long as no significant progress is made on the single relay Gaussian channel, one cannot easily reduce the gap further.
\end{remark}

\begin{remark}[\textbf{Comparison with the two-staged scheme}]
\label{remark:InfiniteVSTwo}
Specializing to the symmetric rate, it can be shown that the \emph{infinite}-staged scheme in Theorem~\ref{theorem:achievableregion} can achieve the symmetric capacity to within 1 bit. Coincidentally, this gap matches the gap result of the \emph{two}-staged scheme in Theorem~\ref{theorem-symmetric}. However, the 1-bit gap comes from different reasons. In the infinite-staged scheme, the 1-bit gap comes from message-splitting. In contrast, in the two-staged scheme, the gap is due to lack of beamforming gain.
One needs to come up with a new technique that well combines these two schemes to reduce the gap into a less-than-one bit.
\end{remark}

\section{The Feedback Capacity Region of El Gamal-Costa Model}
\label{sec:ElGamal-Costa}
We have so far made use of the linear deterministic IC to provide insights into approximating the feedback capacity region of the Gaussian IC. The deterministic IC is a special case of El Gamal-Costa deterministic IC~\cite{ElGamal:it82}. In this section, we establish the exact feedback capacity region for this general class of deterministic ICs.




\begin{figure}[!htp]
\begin{center}
{\epsfig{figure=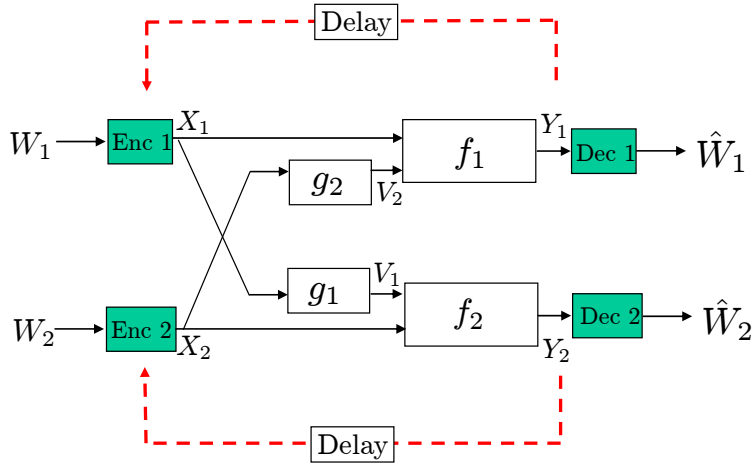, angle=0, width=0.6\textwidth}}
\end{center}
\caption{El Gamal-Costa deterministic IC with feedback} \label{fig_ElGamalCosta}
\end{figure}

Fig.~\ref{fig_ElGamalCosta} $(a)$ illustrates El Gamal-Costa deterministic IC with feedback. The key condition of this model is given by
\begin{align}
\begin{split}
\label{eq-ElGamalCostaCondition}
H(V_2|Y_1,X_1) = 0,\\
H(V_1|Y_2,X_2) = 0,\\
\end{split}
\end{align}
where $V_k$ is a part of $X_k$ ($k=1,2$), visible to the other receiver.
This implies that in any working system where $X_1$ and $X_2$ are decodable at receivers 1 and 2 respectively, $V_1$ and $V_2$ are completely determined at receivers 2 and 1, respectively, i.e., these are common signals.



\begin{theorem}
\label{theorem-ElGamalCosta}
The feedback capacity region of El Gamal-Costa deterministic IC is the set of $(R_1,R_2)$ such that
\begin{align*}
    R_1 &\leq \min \left\{ H(Y_1), H(Y_2|X_2,U) + H(Y_1|V_1,V_2,U) \right\} \\
    R_2 &\leq \min \left\{ H(Y_2), H(Y_1|X_1,U) + H(Y_2|V_1,V_2,U) \right\} \\
    R_1 + R_2 &\leq \min \left\{ H(Y_1|V_1,V_2,U) + H(Y_2), H(Y_2|V_2,V_1,U) + H(Y_1) \right\}
\end{align*}
for some joint distribution $p(u,x_1,x_2)=p(u)p(x_1|u)p(x_2|u)$.
Here $U$ is a discrete random variable which takes on values in the set $\mathcal{U}$ where $|\mathcal{U}| \leq \min (|\mathcal{V}_1||\mathcal{V}_2|, |\mathcal{Y}_1|, |\mathcal{Y}_2|) + 3$.
\end{theorem}
\begin{proof}
Achievability proof is straightforward by Lemma \ref{lemma:feedbackachievable}.
Let $U_k = V_k, \forall k$. Fix a joint distribution $p(u)p(u_1|u)p(u_2|u) p(x_1|u_1,u) p(x_2|u_2,u)$. We now write a joint distribution $p(u,x_1,x_2,u_1,u_2)$ in two different ways:
\begin{align}
\begin{split}
&p(u,x_1,x_2,u_1,u_2) \\
&= p(u)p(x_1|u)p(x_2|u) \delta (u_1 - g_1(x_1)) \delta (u_2-g_2(x_2) \\
&= p(u)p(u_1|u)p(u_2|u) p(x_1|u_1,u) p(x_2|u_2,u)
\end{split}
\end{align}
where $\delta(\cdot)$ indicates the Kronecker delta function. This gives
\begin{align}
\begin{split}
p(x_1|u):= \frac{ p(x_1|u_1,u) p(u_1|u) }{\delta(u_1 - g_1(x_1)) }\\
p(x_2|u):= \frac{ p(x_2|u_2,u) p(u_2|u) }{\delta(u_2 - g_2(x_2)) }
\end{split}
\end{align}
Now we can generate a joint distribution $p(u)p(x_1|u)p(x_2|u)$. Hence, we complete the achievability proof.
See Appendix \ref{Appendix:ElGamalCostaConverse} for converse proof.
\end{proof}

As a by-product, we obtain the feedback capacity region of the linear deterministic IC.
\begin{corollary}
\label{corollary-linear-asymmetric}
The feedback capacity region of the linear deterministic IC is the set of $(R_1,R_2)$ such that
\begin{align*}
    R_1 &\leq \min \left\{ \max(n_{11},n_{12}), \max(n_{11},n_{21})  \right\} \\
    R_2 &\leq \min \left\{ \max(n_{22},n_{21}), \max(n_{22},n_{12})  \right\} \\
    R_1 + R_2 &\leq \min \left\{ \max(n_{22},n_{12}) + (n_{11} -n_{12})^+ , \max(n_{11},n_{21}) + (n_{22} -n_{21})^+  \right\}.
\end{align*}
\end{corollary}
\begin{proof}
The proof is straightforward by Theorem~\ref{theorem-ElGamalCosta}. The capacity region is achieved when $U$ is constant; and $X_1$ and $X_2$ are independent and uniformly distributed.
\end{proof}

\section{Role of Feedback}
\label{sec:feedbackrole}

Recall in Fig.~\ref{fig:gdof} that feedback gain is \emph{bounded} for $0 \leq \alpha \leq \frac{2}{3}$ in terms of the symmetric rate. So a natural question arises: is feedback gain  marginal also from a perspective of the capacity region? With the help of Corollary~\ref{corollary-linear-asymmetric}, we show that feedback can provide \emph{multiplicative} gain even in this regime. We also answer another interesting question posed in Section~\ref{sec:FB_CapacityRegion}: why is the $2R_1+ R_2$ bound missing with feedback?

\subsection{Feedback Gain from a Capacity Region Perspective}
\label{sec:linearDIC}

\begin{figure}[!htp]
\begin{center}
{\epsfig{figure=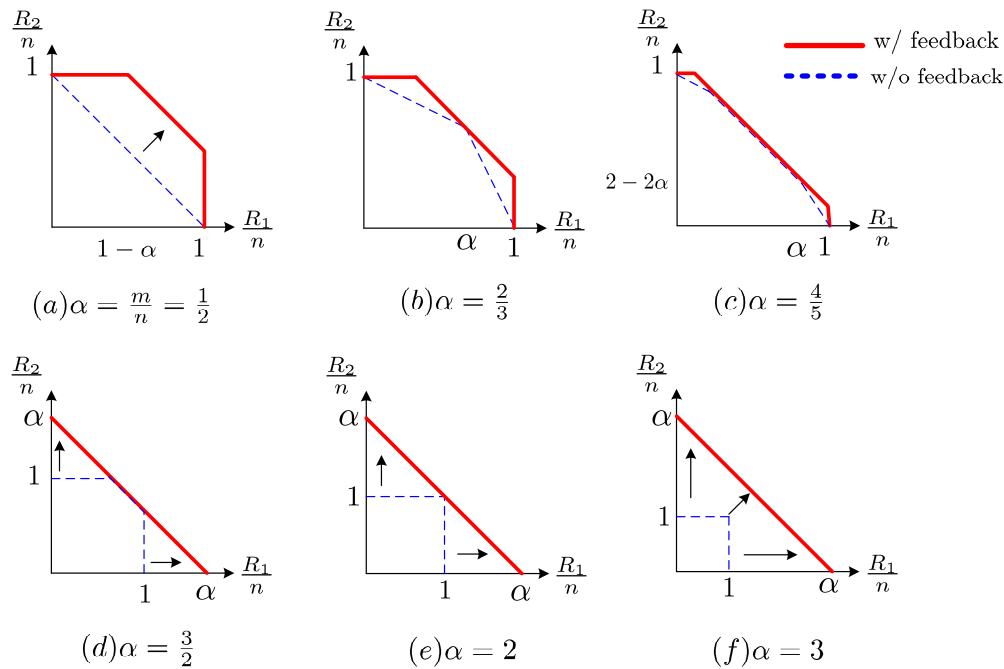, angle=0, width=0.8\textwidth}}
\end{center}
\caption{Feedback capacity region of the linear deterministic IC} \label{fig_LDIC_feedbackgain}
\end{figure}

Fig. \ref{fig_LDIC_feedbackgain} shows the feedback capacity region of the linear deterministic IC under the symmetric channel setting: $n = n_{11}=n_{22}$ and $m = n_{12}=n_{21}$. Interestingly, while for $\frac{2}{3} \leq \alpha \leq 2$, the symmetric capacity does not improve with feedback, the feedback capacity region is enlarged even for this regime.


\subsection{Why is $2R_1 + R_2$ Bound Missing with Feedback?}
\label{sec:2R1R2boundmissing}

Consider an example where $2R_1 + R_2$ bound is active in the non-feedback case. Fig.~\ref{fig_nofeedback_hole} $(a)$ shows an example where a corner point of $(3,0)$ can be achieved. Observe that at the two receivers, the five signal levels are consumed out of the six signal levels. There is one resource hole.
This resource hole is closely related to the $2R_1 + R_2$ bound, which will be shown in Fig.~\ref{fig_nofeedback_hole} $(b)$.

\begin{figure}[htp]
\begin{center}
{\epsfig{figure=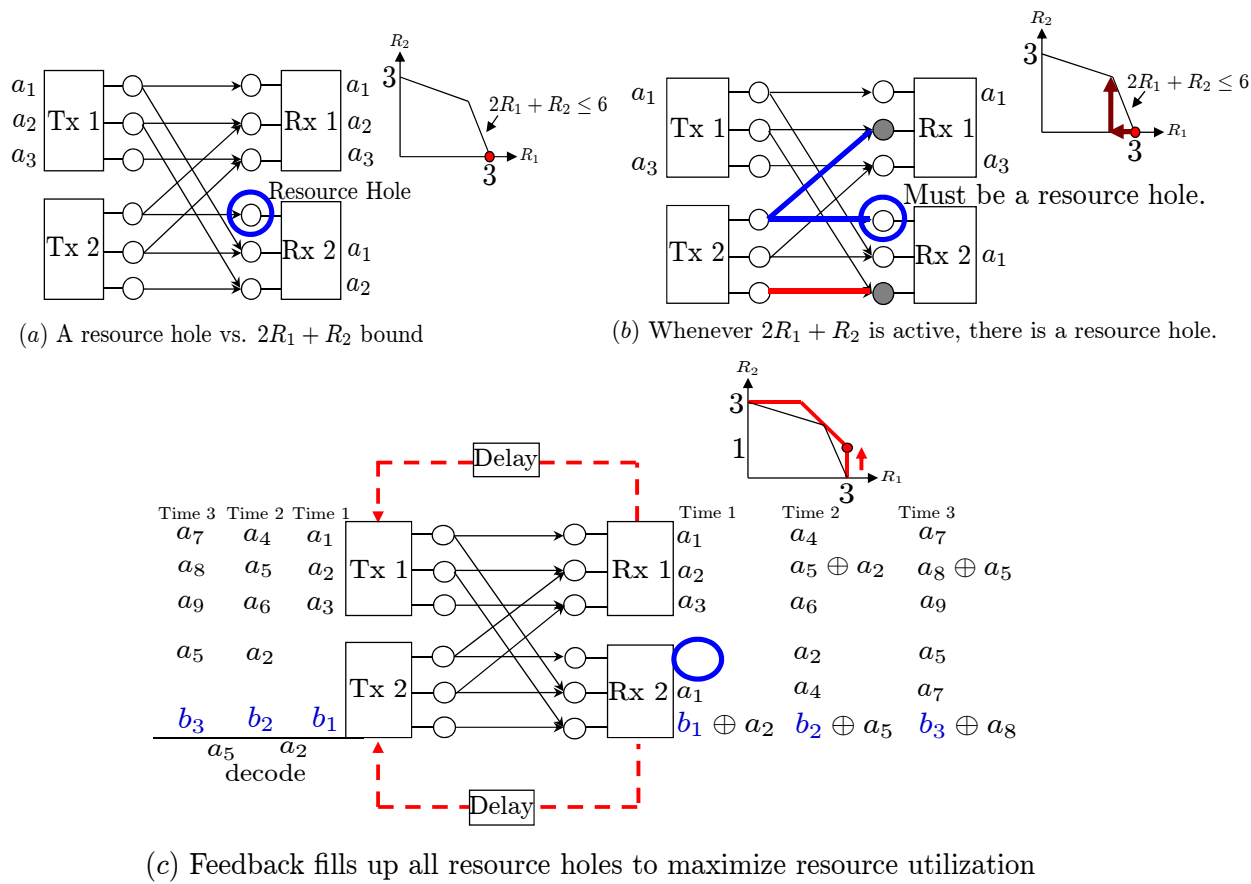, angle=0, width=1.0\textwidth}}
\end{center}
\caption{Relationship between a resource hole and $2R_1+R_2$ bound. The $2R_1 + R_2$ bound is missing with feedback.} \label{fig_nofeedback_hole}
\end{figure}

%
%

Suppose the $2R_1 + R_2$ bound is active. This implies that if $R_1$ is reduced by 1 bit, then $R_2$ should be increased by 2 bits. To decrease $R_1$ by 1 bit, suppose that transmitter 1 sends no information on the second signal level. We then see the two empty signal levels at the two receivers (marked as the gray balls): one at the second level at receiver 1; the other at the bottom level at receiver 2. Now transmitter 2 can send 1 bit on the bottom level to increase $R_2$ by 1 bit (marked as the thick red line). Also it allows transmitter 2 to send one more bit on the top level.
This implies that the top level at receiver 2 must be a resource hole in the previous case. 
This observation combined with the following observation can give an answer to the question.

Fig.~\ref{fig_nofeedback_hole} $(c)$ shows the feedback role that it fills up all the resource holes to maximize resource utilization.
We employ the same feedback strategy used in Fig.~\ref{fig_AchievabilityIdea} to obtain the result in Fig.~\ref{fig_nofeedback_hole} $(c)$. Notice that with feedback, all of the resource holes are filled up except a hole in the first stage, which can be amortized by employing an infinite number of stages. Therefore, we can now see why the $2R_1 + R_2$ bound is missing with feedback.

\section{Discussion}
\label{sec:Discussion}

\subsection{Comparison to Related Work \cite{Kramer:it02, Kramer:it04, GastparKramer:06}}
\label{sec:Comparison}

For the symmetric Gaussian IC, Kramer~\cite{Kramer:it02, Kramer:it04} developed a feedback strategy based on Schalkwijk-Kailath scheme~\cite{SK:it} and Ozarow's scheme~\cite{Ozarow:it}. Due to lack of closed-form rate-formula for the scheme, we cannot see how Kramer's scheme is close to our symmetric rate in Theorem~\ref{theorem-symmetric}. To see this, we compute the generalized degrees-of-freedom of Kramer's scheme.

\begin{lemma}
\label{lemma:KramerGDOF}
The generalized degrees-of-freedom of Kramer's scheme is given by
\begin{align}
\underline{d}(\alpha) = \left\{
                                \begin{array}{ll}
                                  1- \alpha, & \hbox{$0 \leq \alpha < \frac{1}{3}$;} \\
                                  \frac{3-\alpha}{4}, & \hbox{$\frac{1}{3} \leq \alpha < 1$;} \\
                                  \frac{1+\alpha}{4}, & \hbox{$\alpha \geq 1$.}
                                \end{array}
                              \right.
\end{align}
\end{lemma}
\begin{proof}
See Appendix~\ref{appendix:KramerGDOF}.
\end{proof}


\begin{figure}[!htp]
    \begin{minipage}[t]{0.5\linewidth}
        \centering
        \includegraphics[width=3.3in]{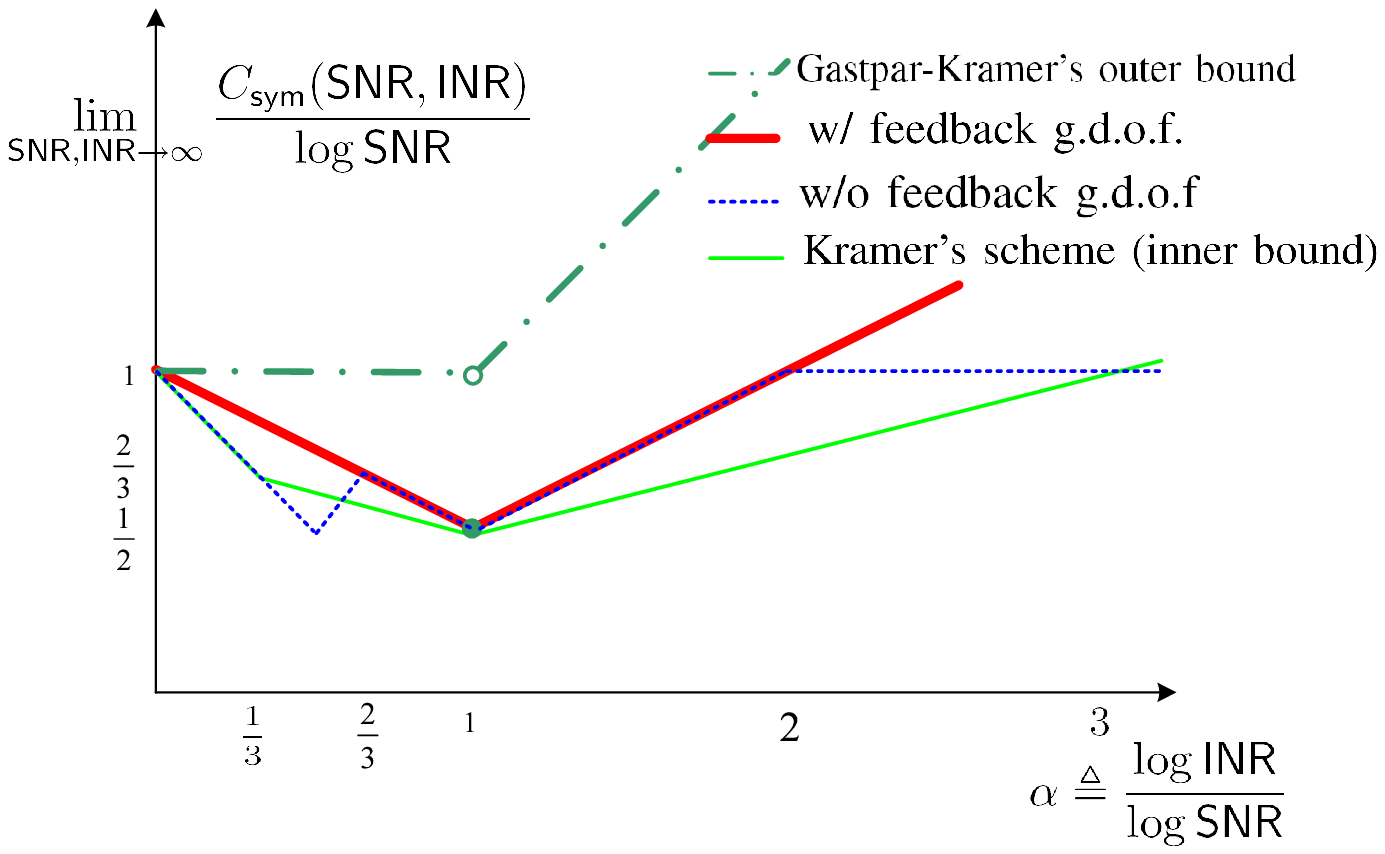}
        \caption{Generalized degrees-of-freedom comparison} \label{fig:gdof-compare}
    \end{minipage}%
    \begin{minipage}[t]{0.5\linewidth}
        \centering
        \includegraphics[width=2.5 in]{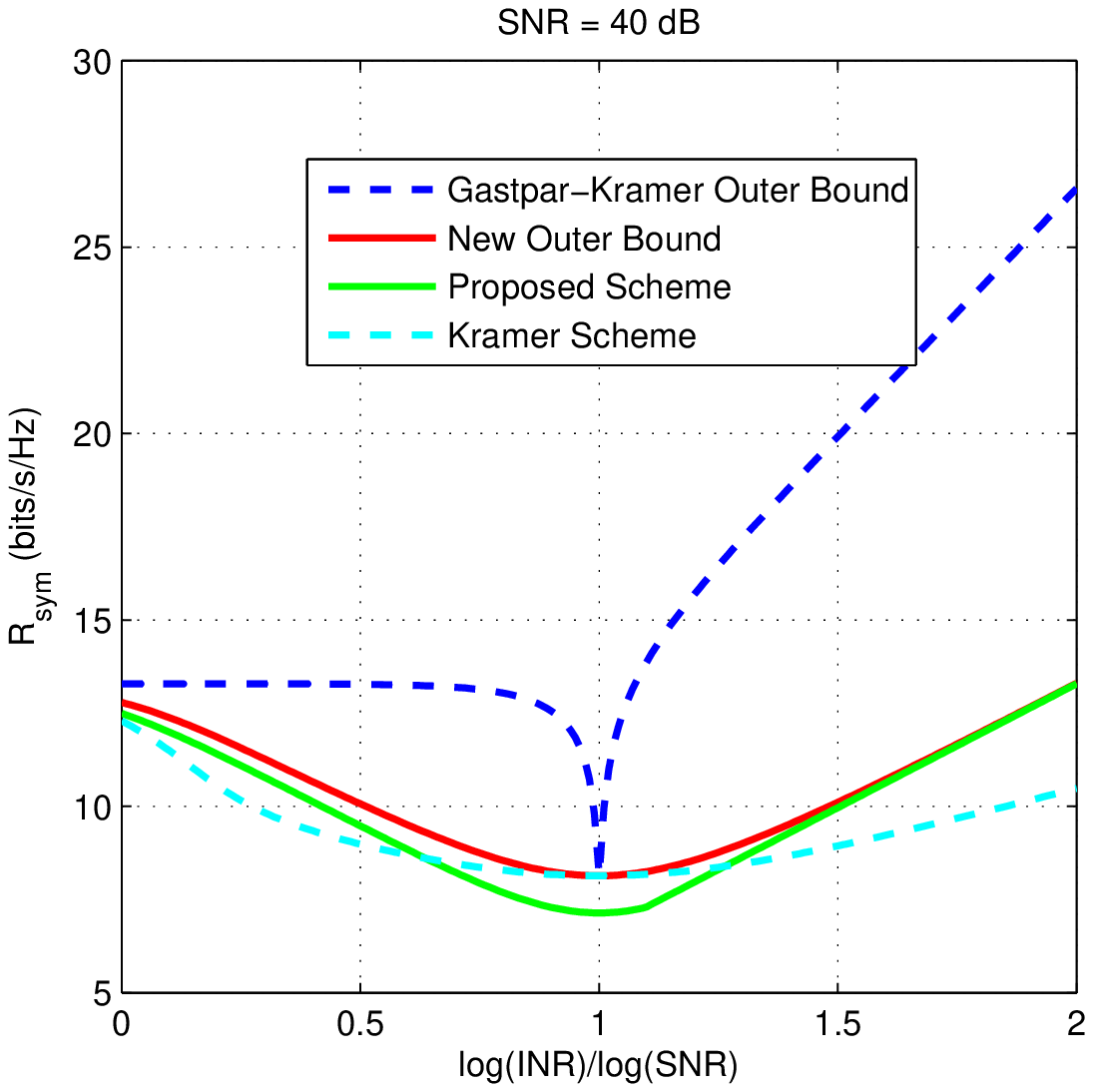}
        \caption{The symmetric rate comparison} \label{fig:ComparisonFinite}
    \end{minipage}
 \end{figure}

Note in Fig. \ref{fig:gdof-compare} that Kramer's scheme can be arbitrarily far from optimality, i.e., it has an unbounded gap to the symmetric capacity for all values of $\alpha$ except $\alpha = 1$. We also plot the symmetric rate for finite channel parameters as shown in Fig. \ref{fig:ComparisonFinite}. Notice that Kramer's scheme is very close to the outer bounds only when $\mathsf{INR}$ is similar to $\mathsf{SNR}$. In fact, the capacity theorem in \cite{Kramer:it04} says that they match each other at $\mathsf{INR} = \mathsf{SNR} - \sqrt{2 \mathsf{SNR}}$. However, if $\mathsf{INR}$ is quite different from $\mathsf{SNR}$, it becomes far away from the outer bounds. Also note that our new bound is much tighter than Gastpar-Kramer's outer bounds in~\cite{Kramer:it02, GastparKramer:06}.


\subsection{Closing the Gap}
\label{sec:closinggap}

\textit{Less than 1-bit gap to the symmetric capacity}: Fig.~\ref{fig:ComparisonFinite} implies that our achievable scheme can be improved especially when $\alpha \approx 1$ where beamforming gain plays a significant role. As mentioned earlier, our two-staged scheme completely loses beamforming gain. In contrast, Kramer's scheme captures the beamforming gain.
As discussed in Remark~\ref{remark:InfiniteVSTwo}, one may develop a unified scheme that beats both the schemes for all channel parameters to reduce the worst-case gap.

\textit{Less than 2-bit gap to the capacity region}:
As mentioned in Remark~\ref{remark:Why2Bits}, a 2-bit gap to the feedback capacity region can be improved up to a 1-bit gap. The idea is to remove message splitting. Recall that the Alamouti-based amplify-and-forward scheme in Theorem~\ref{theorem-symmetric} improves the performance by removing message splitting. Translating the same idea to the characterization of the capacity region is needed for the improvement. A noisy binary expansion model in Fig.~\ref{fig:noisy_binaryexp} may give insights into this.

\subsection{Extension to Gaussian MIMO ICs with Feedback}
The feedback capacity result for El Gamal-Costa model can be extended to Teletar-Tse IC~\cite{Teletar:isit07} where in Fig.~\ref{fig_ElGamalCosta}, $f_k$'s are deterministic functions satisfying El Gamal-Costa condition (\ref{eq-ElGamalCostaCondition}) while $g_k$'s follow arbitrary probability distributions. Once extended, one can infer an approximate feedback capacity region of the two-user Gaussian MIMO IC, as~\cite{Teletar:isit07} did in the non-feedback case.

\section{Conclusion}
\label{sec-Conclusion}

We have established the feedback capacity region to within 2 bits/s/Hz/user and the symmetric capacity to within 1 bit/s/Hz/user universally for the two-user  Gaussian IC with feedback. Alamouti's scheme inspires our two-staged achievable scheme meant for the symmetric rate. For an achievable rate region, we have employed block Markov encoding to incorporate an infinite number of stages. A new outer bound was derived to provide an approximate characterization of the capacity region. As a side-generalization, we have characterized the exact feedback capacity region of El Gamal-Costa deterministic IC.


An interesting consequence of our result is that feedback could provide \emph{multiplicative} gain in many-to-many channels unlike point-to-point, many-to-one, or one-to-many channels. Developing a resource hole interpretation, we have provided qualitative insights as to how feedback can provide significant gain even in the weak interference regime where there is no better alternative path. We have shown that feedback maximizes resource utilization to provide the gain.

\appendices

\section{Achievable Scheme for the Symmetric Rate of~(\ref{eq:SymmetricAchievableRate})}
\label{Appendix:AchiSch_Symmetric}

\begin{figure}[!htp]
\begin{center}
{\epsfig{figure=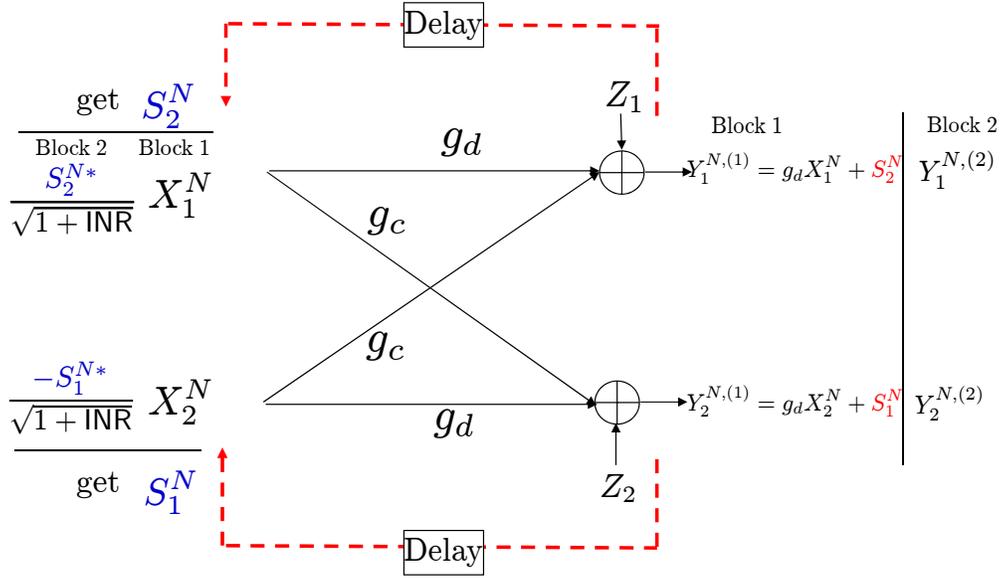, angle=0, width=0.8\textwidth}}
\end{center}
\caption{An achievable scheme in the symmetric Gaussian IC: Alamouti-based amplify-and-forward scheme} \label{fig:Symmetric_Unified}
\end{figure}

The scheme uses two stages (blocks). In the first stage, each transmitter $k$ sends codeword $X_k^{N}$ with rate $R_k$. In the second stage, with feedback transmitter 1 gets the interference plus noise: $S_2^{N} = g_c X_2^{N} + Z_1^{(1),N}$.
Now the complex conjugate technique based on Alamouti's scheme is applied to make $X_1^{N}$ and $S_2^{N}$ well separable.
Transmitters 1 and 2 send $\frac{S_2^{N*}}{\sqrt{1+ \sf INR}}$ and $-\frac{S_1^{N*}}{\sqrt{1+ \sf INR}}$, respectively, where $\sqrt{1+ \sf INR}$ is a normalization factor to meet the power constraint.

Receiver 1 can then gather the two received signals: for $1 \leq i \leq N$,
\begin{align}
\mathbf{Y}_i &\triangleq \left[
  \begin{array}{c}
    Y_{1i}^{(1)} \\
    Y_{1i}^{(2)*} \\
  \end{array}
\right]
= \left[
    \begin{array}{cc}
      g_d & 1 \\
      -\frac{ \mathsf{INR}}{\sqrt{1 + \mathsf{INR}}} & \frac{g_d^{*}}{\sqrt{1 + \mathsf{INR}}} \\
    \end{array}
  \right]
\left[
  \begin{array}{c}
    X_{1i} \\
    S_{2i} \\
  \end{array}
\right]  \quad + \left[
  \begin{array}{c}
    0 \\
    - \frac{g_c^{*}}{\sqrt{1 + \mathsf{INR}}}  Z_{2i}^{(1)} + Z_{1i}^{(2)*}  \\
  \end{array}
\right].
\end{align}
Under Gaussian input distribution, we can compute the rate under MMSE demodulation:
\begin{align}
\frac{1}{2} I(X_{1i}; \mathbf{Y}_i) = \frac{1}{2} h(\mathbf{Y}_i) - \frac{1}{2} h(\mathbf{Y}_i|X_{1i}) = \frac{1}{2} \log \frac{|K_{\mathbf{Y}_i}|}{|K_{\mathbf{Y}_i|X_{1i}}|}.
\end{align}
Straightforward calculations give
\begin{align}
\begin{split}
&|K_{\mathbf{Y}_i}| = \left| \left[
    \begin{array}{cc}
      1+  \mathsf{SNR} + \mathsf{INR} & \frac{g_d}{\sqrt{1 + \mathsf{INR}}} \\
     \frac{g_d^{*}}{\sqrt{1 + \mathsf{INR}}} &
    1+  \mathsf{SNR} + \mathsf{INR}\\
    \end{array}
  \right] \right| = (1+  \mathsf{SNR} + \mathsf{INR})^2 - \frac{\mathsf{SNR}}{1+  \mathsf{INR}}
\\
& |K_{\mathbf{Y}_i|X_{1i}}| = \left| \left[
    \begin{array}{cc}
      1+ \mathsf{INR} & g_d \sqrt{1 + \mathsf{INR}} \\
     g_d^{*} \sqrt{1 + \mathsf{INR}} &
    \mathsf{SNR} + \frac{2\mathsf{INR}+1}{\mathsf{INR}+1}  \\
    \end{array}
  \right] \right| =  1 + 2\mathsf{INR}.
\end{split}
\end{align}
Therefore, we get the desired result: the right term in (\ref{eq:SymmetricAchievableRate}).
\begin{align}
\label{eq:Rsym-impr}
R_{\sf sym} = \frac{1}{2} \log \left( \frac{ (1+ \sf SNR + \sf INR)^2 - \frac{\sf SNR}{1 + \sf INR} }{1 + 2 \sf INR} \right).
\end{align}

\label{sec:noisy_binaryexp}
\begin{figure}[!htp]
\begin{center}
{\epsfig{figure=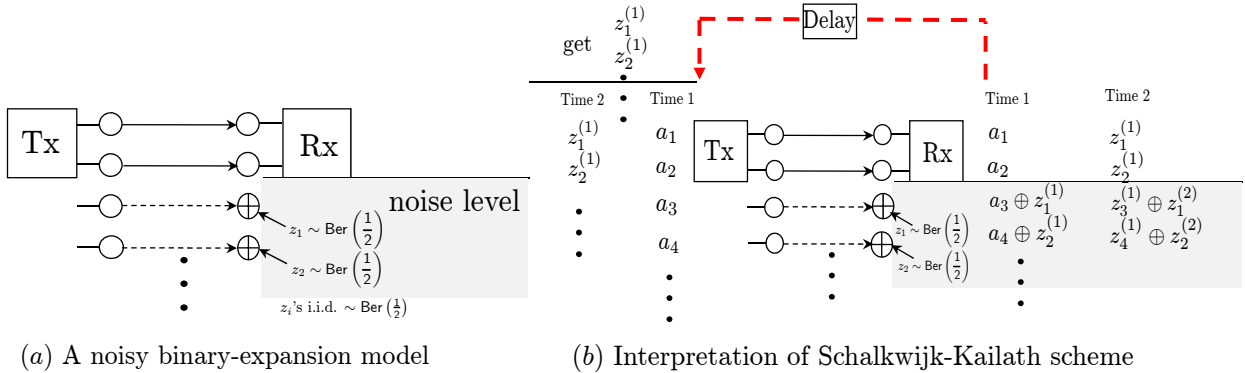, angle=0, width=1.0\textwidth}}
\end{center}
\caption{A noisy binary expansion model. Noise is assumed to be a ${\sf Ber}(\frac{1}{2})$ random variable i.i.d. across time slots (memoryless) and levels. This induces the same capacity as that of the deterministic channel, so it matches the Gaussian channel capacity in the high $\sf SNR$ regime.} \label{fig:noisy_binaryexp}
\end{figure}
\textbf{Intuition Behind the Proposed Scheme:} To provide intuition behind our proposed scheme, we introduce a new model what we call a noisy binary expansion model, illustrated in Fig.~\ref{fig:noisy_binaryexp} $(a)$. In the non-feedback Gaussian channel, due to the absence of noise information at transmitter, transmitter has no chance to refine the corrupted received signal. On the other hand, if feedback is allowed, noise can be learned. Sending noise information (innovation) can refine the corrupted signal: Schalkwijk-Kailath scheme \cite{SK:it}. However, a linear deterministic model cannot capture interplay between noise and signal. To capture this issue, we slightly modify the deterministic model so as to reflect the effect of noise. In this model, we assume that noise is a ${\sf Ber}(\frac{1}{2})$ random variable i.i.d. across time slots (memoryless) and levels. This induces the same capacity as that of the deterministic channel, so it matches the Gaussian channel capacity in the high $\sf SNR$ regime.

As a stepping stone towards the interpretation of the proposed scheme, let us first understand Schalkwijk-Kailath scheme~\cite{SK:it} using this model. Fig.~\ref{fig:noisy_binaryexp} $(b)$ illustrates an example where 2 bits/time can be sent with feedback. In time 1, transmitter sends independent bit streams $(a_1,a_2,a_3,a_4, \cdots) $. Receiver then gets $(a_1, a_2, a_3 \oplus z_1^{(1)}, a_4 \oplus z_2^{(1)}, \cdots)$ where $z_i^{(j)}$ indicates an i.i.d. ${\sf Ber}\left( \frac{1}{2} \right) $ random variable of noise level $i$ at time $j$. With feedback, transmitter can get noise information $(0, 0, z_1^{(1)}, z_2^{(1)}, \cdots )$ by subtracting the transmitted signals (sent previously) from the received feedback. This process corresponds to an MMSE operation in Schalkwijk-Kailath scheme: computing innovation. Transmitter scales the noise information to shift it by 2 levels and then sends the shifted version. The shifting operation corresponds to a scaling operation in Schalkwijk-Kailath scheme. Receiver can now recover $(a_3,a_4)$ corrupted by $(z_1^{(1)}, z_2^{(1)})$ in the previous slot. We repeat this procedure.

\begin{figure}[t]
\begin{center}
{\epsfig{figure=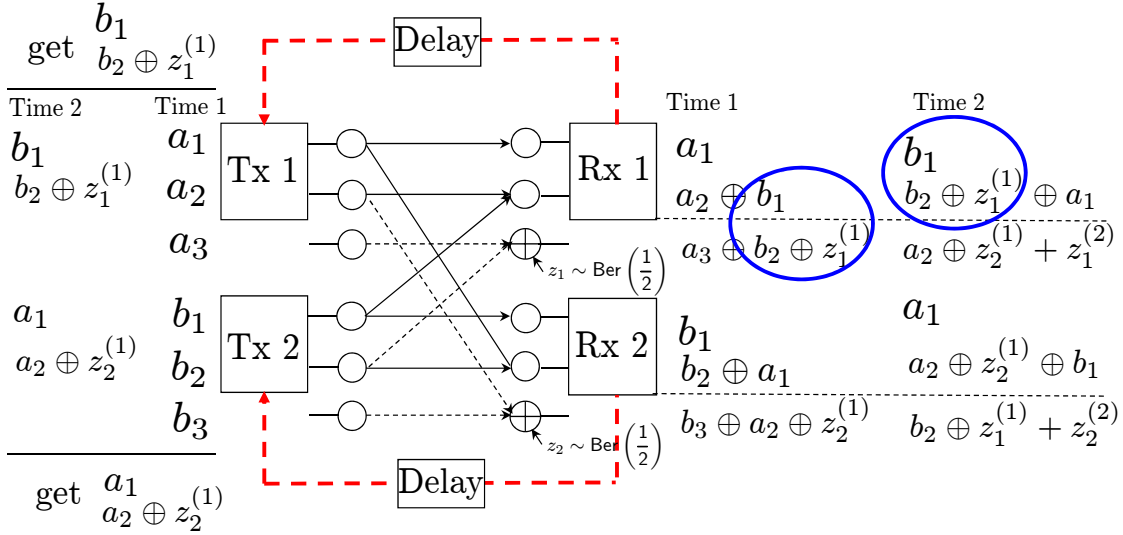, angle=0, width=0.9\textwidth}}
\end{center}
\caption{Intuition behind the Alamouti-based amplify-and-forward scheme.} \label{fig:Symmetric_Unified_BinaryModel}
\end{figure}

The viewpoint based on a binary expansion model can provide intuition behind our proposed scheme. See Fig.~\ref{fig:Symmetric_Unified_BinaryModel}.
%
In the first stage, each transmitter sends three independent bits: two bits above the noise level; one bit below the noise level. Transmitters 1 and 2 send $(a_1,a_2,a_3)$ and $(b_1,b_2,b_3)$, respectively. Receiver 1 then gets: (1) the clean signal $a_1$; (2) the interfered signal $a_2 \oplus b_1$; and (3) the interfered-and-noised signal $a_3 \oplus b_2 \oplus z_1^{(1)}$. Similarly for receiver 2. In the second stage, with feedback, each transmitter can get interference plus noise by subtracting the transmitted signals from the feedback. Transmitters 1 and 2 get $(0,b_1,b_2 \oplus z_1^{(1)})$ and $(0,a_1,a_2 \oplus z_2^{(1)})$, respectively. Next, each transmitter scales the subtracted signal subject to the power constraint and then forwards the scaled signal. Transmitters 1 and 2 send $(b_1,b_2 \oplus z_1^{(1)})$ and $(a_1,a_2 \oplus z_2^{(1)})$, respectively. Each receiver can then gather the two received signals to decode 3 bits. From this figure, one can see that it is not needed to send additional information on top of innovation in the second stage. Therefore, this scheme matches Alamouti-based amplify-and-forward scheme in the Gaussian channel.

%


\section{Proof of Lemma \ref{lemma:feedbackachievable}}
\label{Appendix:lemmaachievable}

\textbf{Codebook Generation:} Fix a joint distribution $p(u)p(u_1|u) p(u_2|u) p(x_1|u_1, u)p(x_2|u_2,u)$.
First generate $2^{N(R_{1c} +R_{2c}) }$ independent codewords $u^N(i,j)$, $i \in \{1, \cdots, 2^{NR_{1c}} \}$, $j\in \{1, \cdots, 2^{NR_{2c}} \}$, according to $\prod_{i=1}^{N} p(u_i)$. For each codeword $u^N(i,j)$, encoder 1 generates $2^{NR_{1c}}$ independent codewords $u_1^{N}((i,j),k)$, $k \in \{1, \cdots, 2^{NR_{1c}} \}$, according to $\prod_{i=1}^{N} p(u_{1i}|u_i)$. Subsequently, for each pair of codewords $\left( u^N(i,j),u_1^N((i,j),k) \right) $, generate $2^{NR_{1p}}$ independent codewords $x_1^N((i,j),k,l)$, $l \in \{1, \cdots, 2^{NR_{1p}} \}$, according to $\prod_{i=1}^{N} p(x_{1i}|u_{1i},u_i)$.

Similarly, for each codeword $u^{N}(i,j)$, encoder 2 generates $2^{NR_{2c}}$ independent codewords $u_2^{N}((i,j),r)$, $r \in \{1, \cdots, 2^{NR_{2c}} \}$, according to $\prod_{i=1}^{N} p(u_{2i}|u_i)$. For $ \left( u^N(i,j),u_2^N((i,j),r) \right)$, generate $2^{NR_{2p}}$ independent codewords $x_2^N((i,j),r,s)$, $s \in \{1, \cdots, 2^{NR_{2p}} \}$, according to $\prod_{i=1}^{N} p(x_{2i}|u_{2i},u_i)$.

\textit{Notation:} Notations are independently used only for this section. The index $k$ indicates the common message of user 1 instead of user index. The index $i$ is used for both purposes: (1) indicating the previous common message of user 1; (2) indicating time index. It could be easily differentiated from contexts.

\textbf{Encoding and Decoding:} We employ block Markov encoding with a total size $B$ of blocks. Focus on the $b$th block transmission. With feedback $y_1^{N,(b-1)}$, transmitter 1 tries to decode the message $\hat{w}_{2c}^{(b-1)} = \hat{k}$ (sent from transmitter 2 in the $(b-1)$th block). In other words, we find the unique $\hat{k}$ such that
\begin{align*}
\begin{split}
&\left( u^N \left( w_{1c}^{(b-2)}, \hat{w}_{2c}^{(b-2)} \right), u_1^N \left( ( w_{1c}^{(b-2)}, \hat{w}_{2c}^{(b-2)}), w_{1c}^{(b-1)} \right), \right. \\
&\; \; \left. x_1^N \left( (w_{1c}^{(b-2)}, \hat{w}_{2c}^{(b-2)} ), w_{1c}^{(b-1)},   w_{1p}^{(b-1)} \right), \right. \\
& \left. \;\; u_2^N \left( (w_{1c}^{(b-2)}, \hat{w}_{2c}^{(b-2)}), \hat{k} \right) , y_1^{N,(b-1)}  \right) \in A_{\epsilon}^{(N)},
\end{split}
\end{align*}
where $A_{\epsilon}^{(N)}$ indicates the set of jointly typical sequences.
Note that transmitter 1 already knows its own messages $(w_{1c}^{(b-2)}, w_{1c}^{(b-1)}, w_{1p}^{(b-1)})$. We assume that $\hat{w}_{2c}^{(b-2)}$ is correctly decoded from the previous block $(b-1)$. The decoding error occurs if one of two events happens: (1) there is no typical sequence; (2) there is another $\hat{w}_{2c}^{(b-1)}$ such that it is a typical sequence. By AEP, the first error probability becomes negligible as $N$ goes to infinity. By \cite{CoverThomas}, the second error probability becomes arbitrarily small (as $N$ goes to infinity) if
\begin{align}
\label{eq-R2c-constraint}
R_{2c} \leq I(U_2;Y_1|X_1,U).
\end{align}
Based on $(w_{1c}^{(b-1)},\hat{w}_{2c}^{(b-1)})$, transmitter 1 generates a new common message $w_{1c}^{(b)}$ and a private message $w_{1p}^{(b)}$. It then sends $x_1^N \left( (w_{1c}^{(b-1)}, \hat{w}_{2c}^{(b-1)} ),w_{1c}^{(b)},w_{1p}^{(b)} \right)$. Similarly transmitter 2 decodes $\hat{w}_{1c}^{(b-1)}$, generates $(w_{2c}^{(b)},w_{2p}^{(b)})$ and then sends $x_2^N \left((\hat{w}_{1c}^{(b-1)}, w_{2c}^{(b-1)} ),w_{2c}^{(b)},w_{2p}^{(b)} \right)$.

Each receiver waits until total $B$ blocks have been received and then does \emph{backward decoding}.
Notice that a block index $b$ starts from the last $B$ and ends to $1$. For block $b$, receiver 1 finds the unique triple $(\hat{i},\hat{j},\hat{k})$ such that
\begin{align*}
\begin{split}
\left( u^N \left( \hat{i}, \hat{j} \right), u_1^N \left( ( \hat{i}, \hat{j}) , \hat{w}_{1c}^{(b)} \right),  x_1^N \left( (\hat{i}, \hat{j}), \hat{w}_{1c}^{(b)},   \hat{k} \right), \right. \\
 \left. \;\; u_2^N \left( (\hat{i}, \hat{j}), \hat{w}_{2c}^{(b)} \right) , y_1^{N,(b)}  \right) \in A_{\epsilon}^{(N)},
\end{split}
\end{align*}
where we assumed that a pair of messages $(\hat{w}_{1c}^{(b)}, \hat{w}_{2c}^{(b)})$ was successively decoded from block $(b+1)$. Similarly receiver 2 decodes $(\hat{w}_{1c}^{(b-1)}, \hat{w}_{2c}^{(b-1)},\hat{w}_{2p}^{(b)})$.

\textbf{Error Probability:} By symmetry, we consider the probability of error only for block $b$ and for a pair of transmitter 1 and receiver 1.
We assume that $(w_{1c}^{(b-1)},w_{2c}^{(b-1)}, w_{1p}^{(b)}) = (1,1,1)$ was sent through block $(b-1)$ and block $b$; and there was no backward decoding error from block $B$ to $(b+1)$, i.e., $(\hat{w}_{1c}^{(b)}, \hat{w}_{2c}^{(b)})$ are successfully decoded.

Define an event:
\begin{align*}
E_{ijk} = \left\{ \left( u^{N}(i,j), u_1^{N}((i,j),\hat{w}_{1c}^{(b)}), x_1^{N}( (i,j),\hat{w}_{1c}^{(b)},k ) , \right. \right. \\
\left. \left. \;\; u_2^{N}((i,j),\hat{w}_{2c}^{(b)}), y_1^{N,(b)} \right)  \in A_{\epsilon}^{(N)} \right\}.
\end{align*}
By AEP, the first type of error becomes negligible. Hence, we focus only on the second type of error. Using the union bound, we get
\begin{align}
\begin{split}
\label{eq-errorprobability}
\textrm{Pr} & \left( \bigcup_{(i,j,k)  \neq (1,1,1) } E_{ijk}   \right) \leq \sum_{i \neq 1, j \neq 1, k \neq 1}\textrm{Pr}(E_{ijk} ) + \sum_{i \neq 1, j \neq 1, k =1}\textrm{Pr}(E_{ij1} )  + \sum_{i \neq 1, j = 1, k \neq 1}\textrm{Pr}(E_{i1k} ) \\
&\quad + \sum_{i \neq 1, j = 1, k =1}\textrm{Pr}(E_{i11} ) + \sum_{i = 1, j \neq 1, k \neq 1}\textrm{Pr}(E_{1jk} ) + \sum_{i = 1, j \neq 1, k =1}\textrm{Pr}(E_{1j1} ) + \sum_{i = 1, j = 1, k \neq 1}\textrm{Pr}(E_{11k} ) \\
& \leq 2^{N(R_{1c}+R_{2c}+R_{1p} - I(U,X_1,U_2;Y_1)+ 4 \epsilon)} +2^{N(R_{1c}+R_{2c} - I(U,X_1,U_2;Y_1)+ 4 \epsilon)} +2^{N(R_{1c}+R_{1p} - I(U,X_1,U_2;Y_1)+ 4 \epsilon)} \\
 & \; +2^{N(R_{1c} - I(U,X_1,U_2;Y_1)+ 4 \epsilon)}  +2^{N(R_{2c} + R_{1p} - I(U,X_1,U_2;Y_1)+ 4 \epsilon)}  +2^{N(R_{2c} - I(U,X_1,U_2;Y_1)+ 4 \epsilon)} \\
 & \; +2^{N(R_{1p} - I(X_1;Y_1|U,U_1,U_2)+ 4 \epsilon)}.
\end{split}
\end{align}

From (\ref{eq-R2c-constraint}) and (\ref{eq-errorprobability}), we can say that the error probability can be made arbitrarily small if
\begin{align}
&\left\{
  \begin{array}{ll}
    R_{2c} &\leq I (U_2; Y_1| X_1 ,U) \\
    R_{1p} &\leq I (X_1; Y_1| U_1, U_2,U) \\
    R_{1c}+R_{1p}+R_{2c} &\leq I(U, X_1, U_2; Y_1)
  \end{array}
\right. \\
&\left\{
  \begin{array}{ll}
    R_{1c} &\leq I (U_1; Y_2| X_2 ,U) \\
    R_{2p} &\leq I (X_2; Y_2| U_1, U_2,U) \\
    R_{2c}+R_{2p}+R_{1c} &\leq I(U, X_2, U_1; Y_2).
  \end{array}
\right.
\end{align}

\textbf{Fourier-Motzkin Elimination:} Applying Fourier-Motzkin elimination, we easily obtain the desired inequalities.
There are several steps to remove $R_{1p}$, $R_{2p}$, $R_{1c}$, and $R_{2c}$, successively.
First substitute $R_{1p} = R_1 - R_{1c}$ and $R_{2p} = R_2 - R_{2c}$ to get:
\begin{align}
R_{2c} &\leq I (U_2; Y_1| X_1 ,U) &:= a_1\\
R_1 - R_{1c} &\leq I (X_1; Y_1| U_1, U_2,U) &:= a_2\\
 R_1 +R_{2c} &\leq I(U, X_1, U_2; Y_1) &:= a_3\\
R_{1c} &\leq I (U_1; Y_2| X_2 ,U) &:= b_1 \\
 R_2 - R_{2c} &\leq I (X_2; Y_2| U_1, U_2,U) &:= b_2\\
R_2 +R_{1c} &\leq I(U, X_2, U_1; Y_2) &:= b_3 \\
-R_{1c} & \leq 0 \\
-R_1 + R_{1c} & \leq 0 \\
-R_{2c} & \leq 0 \\
-R_2 + R_{2c} & \leq 0
\end{align}

Categorize the above inequalities into the following three groups: (1) group 1 not containing $R_{1c}$; (2) group 2 containing \emph{negative} $R_{1c}$; (3) group 3 containing \emph{positive} $R_{1c}$. By adding each inequality from groups 2 and 3, we remove $R_{1c}$. Rearranging the inequalities with respect to $R_{2c}$, we get:


%
\begin{align}
R_{1} &\leq  b_1 + a_2\\
R_2 +R_{1} &\leq b_5  + a_2\\
-R_1  & \leq 0 \\
R_{2c} &\leq  a_1\\
R_1 +R_{2c} &\leq  a_5\\
-R_2 + R_{2c} & \leq 0 \\
R_2 - R_{2c} &\leq  b_2\\
 -R_{2c} & \leq 0.
\end{align}

Adding each inequality from groups 2 and 3, we remove $R_{2c}$ and finally obtain:
%
\begin{align}
R_{1} &\leq  \min(a_5, b_1 + a_2)\\
R_{2} &\leq  \min(b_5, a_1 + b_2)\\
R_1 +R_{2} &\leq \min(b_5  + a_2, a_5 + b_2).
\end{align}

\section{Converse Proof of Theorem~\ref{theorem-ElGamalCosta}}
\label{Appendix:ElGamalCostaConverse}

For completeness, we provide the detailed proof, although there are many overlaps with the proof in Theorem~\ref{theorem:outerbound}. The main point of the converse is how to introduce an auxiliary random variable $U$ which satisfies that given $U_i$, $X_{1i}$ is conditionally independent of $X_{2i}$. Claim~\ref{claim-3} gives hint into this. It gives the choice of $U_i:=(V_1^{i-1},V_2^{i-1})$.

First we consider the upper bound of an individual rate.
\begin{align*}
\begin{split}
NR_1 &= H(W_1) \overset{(a)}{\leq} I(W_1;Y_1^{N}) +  N\epsilon_N \overset{(b)}{\leq} \sum  H(Y_{1i}) + N\epsilon_N\\
\end{split}
\end{align*}
where $(a)$ follows from Fano's inequality and $(b)$ follows from the fact that entropy is non-negative and conditioning reduces entropy.

Now consider the second bound.
\begin{align*}
\begin{split}
N&R_1 = H(W_1) = H(W_1|W_2) \\
&\leq I(W_1;Y_1^N|W_2) + N \epsilon_N \leq I(W_1;Y_1^N,Y_2^{N}|W_2) + N \epsilon_N \\
&\overset{(a)}= \sum H(Y_{1i},Y_{2i}|W_2,Y_1^{i-1},Y_2^{i-1}) + N \epsilon_N \\
&\overset{(b)}{=} \sum H(Y_{1i},Y_{2i}|W_2, Y_1^{i-1},Y_2^{i-1},X_{2}^{i}) + N \epsilon_N \\
&\overset{(c)}{=} \sum H(Y_{2i}|W_2,Y_1^{i-1},Y_2^{i-1}, X_{2}^{i}) \\
& \quad + \sum H(Y_{1i}|W_2, Y_1^{i-1},Y_2^{i-1}, X_{2}^{i},Y_{2i},V_{1}^{i}) + N \epsilon_N \\
&\overset{(d)}{\leq} \sum \left[ H(Y_{2i}|X_{2i}, U_i) + H(Y_{1i}|V_{1i},V_{2i},U_i) \right] + N \epsilon_N
\end{split}
\end{align*}
where ($a$) follows from the fact that $(Y_1^{N},Y_2^{N})$ is a function of $(W_1,W_2)$; $(b)$ follows from the fact that $X_{2}^{i}$ is a function of $(W_2,Y_2^{i-1})$; ($c$) follows from the fact that $V_{1}^{i}$ is a function of $(Y_{2}^{i},X_{2}^{i})$; ($d$) follows from the fact that $V_1^{i-1}$ is a function of $(Y_{2}^{i-1},X_2^{i-1})$, $V_2^{i-1}$is a function of $X_2^{i-1}$, and conditioning reduces entropy. Similarly we get the outer bound for $R_2$.

The sum rate bound is given as follows.
\begin{align*}
\begin{split}
N&(R_1 + R_2)= H(W_1) + H(W_2) =  H(W_1|W_2)+ H(W_2) \\
&\leq I(W_1;Y_1^{N}|W_2) + I(W_2;Y_2^{N}) + N \epsilon_N \\
& = H(Y_1^{N}|W_2) + I(W_2;Y_2^{N}) + N \epsilon_N \\
&= H(Y_1^{N}|W_2) + H(Y_2^{N}) \\
&\quad - \left\{ H(Y_1^{N}, Y_2^{N}|W_2) - H(Y_1^{N}|Y_2^{N},W_2) \right\} + N \epsilon_N \\
& = H(Y_1^{N}|Y_2^{N},W_2) - H(Y_2^{N}|Y_1^{N},W_2) + H(Y_2^N) + N \epsilon_N \\
&\overset{(a)}{=} \sum H(Y_{1i} | Y_1^{i-1}, Y_2^{N}, W_2, X_2^{i},V_1^i) + H(Y_2^N) +  N \epsilon_N \\
&\overset{(b)}{\leq} \sum \left[ H(Y_{1i}|V_{1i},V_{2i},U_i) + H(Y_{2i}) \right] + N \epsilon_N
\end{split}
\end{align*}
where ($a$) follows from the fact that $X_2^{i}$ is a function of $(W_2,Y_2^{i-1})$ and $V_1^{i}$ is a function of $(X_2^{i},Y_2^{i})$; ($b$) follows from the fact that $V_2^i$ is a function of $X_2^i$ and conditioning reduces entropy. Similarly, we get the other outer bound:
\begin{align*}
N(R_1 + R_2) \leq \sum \left[ H(Y_{2i}|V_{1i},V_{2i},U_i) + H(Y_{1i}) \right] + N \epsilon_N.
\end{align*}

Now let a time index $Q$ be a random variable uniformly distributed over the set $\{1,2,\cdots,N\}$ and independent of $(W_1,W_2,X_1^N,X_2^N,Y_1^N,Y_2^N)$.
We define
\begin{align}
\begin{split}
&X_1 = X_{1Q}, \;V_1 = V_{1Q}; \;X_2 = X_{2Q}, \;V_2 = V_{1Q}, \\
&Y_1 = Y_{1Q}, \;Y_2 = Y_{2Q}; \;U = (U_Q, Q).
\end{split}
\end{align}
If $(R_1,R_2)$ is achievable, then $\epsilon_N \rightarrow 0$ as $N \rightarrow \infty$.
By Claim \ref{claim-3}, an input joint distribution satisfies $p(u,x_1,x_2)=p(u)p(x_1|u)p(x_2|u)$.
This establishes the converse.

\begin{claim}
\label{claim-3}
Given $U_i = (V_1^{i-1},V_2^{i-1})$, $X_{1i}$ and $X_{2i}$ are conditionally independent.
\end{claim}
\begin{proof}
The proof is based on the dependence-balance-bound technique in \cite{Willems:it82, Willems:it89}.
For completeness we describe details. First we show that $I(W_1;W_2|U_i)=0$, which implies that $W_1$ and $W_2$ are independent given $U_i$. Based on this, we show that $X_{1i}$ and $X_{2i}$ are conditionally independent given $U_i$.

Consider
\begin{align*}
0 &\leq I(W_1;W_2|U_i) \overset{(a)}{=} I(W_1;W_2|U_i) - I(W_1;W_2) \\
& \overset{(b)}{=} - H(W_1) - H(W_2) - H(U_i) + H(W_1, W_2) \\
& \;\; + H(W_1,U_i) + H(W_2,U_i) - H(W_1,W_2, U_i) \\
& \overset{(c)}{=} - H(U_i) + H(U_i|W_1) + H(U_i|W_2) \\
& = \sum_{j=1}^{i-1} \left[ - H(V_{1j},V_{2j} |V_1^{j-1}, V_2^{j-1})   \right. \\
&  \quad \qquad + H(V_{1j},V_{2j}|W_1, V_1^{j-1}, V_2^{j-1}) \\
& \quad \qquad \left. +   H(V_{1j},V_{2j} |W_2, V_1^{j-1}, V_2^{j-1}) \right] \\
& \overset{(d)}{=} \sum_{j=1}^{i-1} \left[ - H(V_{1j},V_{2j} |V_1^{j-1}, V_2^{j-1})  \right. \\
& \left. \qquad + H(V_{2j}|W_1, V_1^{j}, V_2^{j-1}) + H(V_{1j}|W_2, V_1^{j-1}, V_2^{j} ) \right] \\
& = \sum_{j=1}^{i-1} \left[ - H(V_{1j} |V_1^{j-1}, V_2^{j-1} ) + H(V_{1j} |W_2, V_1^{j-1}, V_2^{j}) \right. \\
&\left. \qquad - H(V_{2j}|V_1^{j}, V_2^{j-1}) + H(V_{2j}|W_1, V_1^{j}, V_2^{j-1}) \right] \\
& \overset{(e)} \leq 0
\end{align*}
where ($a$) follows from $I(W_1;W_2)=0$; ($b$) follows from the chain rule; ($c$) follows from the chain rule and $H(U_i|W_1,W_2)=0$; ($d$) follows from the fact that $V_1^{j}$ is a function of $(W_1,V_2^{j-1})$ and $V_2^{j}$ is a function of $(W_2,V_1^{j-1})$ (see Claim \ref{claim-1}); ($e$) follows from the fact that conditioning reduces entropy.
Therefore, $I(W_1;W_2|U_i)=0$, which shows the independence of $W_1$ and $W_2$ given $u_i$.

Notice that $X_{1i}$ is a function of $(W_1,V_2^{i-1})$ and $X_{2i}$ is a function of $(W_2,V_1^{i-1})$ (see Claim \ref{claim-1}). Hence, it follows easily that
\begin{align}
I(X_{1i};X_{2i}|U_i) = I(X_{1i};X_{2i}|V_{1}^{i-1}, V_2^{i-1} ) = 0,
\end{align}
which proves the independence of $X_{1i}$ and $X_{2i}$ given $U_i$.
\end{proof}

\begin{claim}
\label{claim-1}
For $i\geq 1$, $X_1^{i}$ is a function of $(W_1,V_2^{i-1})$. Similarly, $X_2^{i}$ is a function of $(W_2,V_1^{i-1})$.
\end{claim}
\begin{proof}
By symmetry, it is enough to prove it only for $X_1^{i}$.
Since the channel is deterministic (noiseless), $X_1^{i}$ is a function of $(W_1, W_2)$. In Fig. \ref{fig_ElGamalCosta}, we see that information of $W_2$ to the first link pair must pass through $V_{2i}$. Also note that $X_{1i}$ depends on the past output sequences until $i-1$ (due to feedback delay). Therefore, $X_1^{i}$ is a function of $(W_1,V_2^{i-1})$.
\end{proof}

\section{Proof of Lemma~\ref{lemma:KramerGDOF}}
\label{appendix:KramerGDOF}

Let $\mathsf{INR} = \mathsf{SNR}^{\alpha}$. Then, by (29) in \cite{Kramer:it02} and (77*) in \cite{Kramer:it04}, we get
\begin{align}
R_{\mathsf{sym}} = \log \left( \frac{1 + \mathsf{SNR} + \mathsf{SNR}^{\alpha} + 2\rho^{*} \mathsf{SNR}^{\frac{\alpha+1}{2}} }{ 1 + (1-\rho^{*2}) \mathsf{SNR}^{\alpha}}  \right),
\end{align}
where $\rho^{*}$ is the solution between 0 and 1 such that
\begin{align*}
&2 \mathsf{SNR}^{\frac{3 \alpha +1 }{2}} \rho^{*4} +  \mathsf{SNR}^{ \alpha } \rho^{*3} - 4 ( \mathsf{SNR}^{\frac{3 \alpha +1 }{2}} + \mathsf{SNR}^{\frac{\alpha +1 }{2}} ) \rho^{*2} \\
&- ( 2 + \mathsf{SNR} + 2 \mathsf{SNR}^{ \alpha } ) \rho^{*} + 2 ( \mathsf{SNR}^{\frac{3 \alpha +1 }{2}} + \mathsf{SNR}^{\frac{\alpha +1 }{2}} ) = 0.
\end{align*}
Notice that for $0 \leq \alpha \leq \frac{1}{3}$ and for the high $\sf SNR$ regime, $\mathsf{SNR}$ is a dominant term and  $0 < \rho^{*} <1 $. Hence, we get $\rho^{*} \approx 2 \mathsf{SNR}^{\frac{3 \alpha -1 }{2}}$. This gives $\lim_{\mathsf{SNR} \rightarrow \infty} \frac{R_{\mathsf{sym}}}{\log(\mathsf{SNR})}  = 1 - \alpha$. For $\frac{1}{3} <  \alpha < 1$, the first and second dominant terms become $\mathsf{SNR}^{\frac{3 \alpha +1}{2}}$ and $\mathsf{SNR}$, respectively. Also for this regime, $\rho^{*} \approx 1$. Hence, we approximately get $1-\rho^{*2} \approx \mathsf{SNR}^{\frac{-3 \alpha+1}{4}}$. This gives $\lim_{\mathsf{SNR} \rightarrow \infty} \frac{R_{\mathsf{sym}}}{\log(\mathsf{SNR})}  =  \frac{ 3 - \alpha}{4}$.
For $\alpha \geq 1$, note that the first and second dominant terms are $\mathsf{SNR}^{\frac{3 \alpha +1}{2}}$ and $\mathsf{SNR}$; and $\rho^{*}$ is very close to 1. So we get $1-\rho^{*2} \approx \mathsf{SNR}^{-\frac{ \alpha+1}{4}}$. This gives the desired result in the last case.

\bibliographystyle{ieeetr}
\bibliography{bib_feedback}

\begin{thebibliography}{10}

\bibitem{shannon:it}
C.~E. Shannon, ``The zero error capacity of a noisy channel,'' {\em IRE
  Transactions on Information Theory}, Sept. 1956.

\bibitem{Cover:it89}
T.~M. Cover and S.~Pombra, ``Gaussian feedback capacity,'' {\em IEEE
  Transactions on Information Theory}, vol.~35, pp.~37--43, Jan. 1989.

\bibitem{Kim:it06}
Y.-H. Kim, ``Feedback capacity of the first-order moving average {G}aussian
  channel,'' {\em IEEE Transactions on Information Theory}, vol.~52,
  pp.~3063--3079, July 2006.

\bibitem{Gaarder:it}
N.~T. Gaarder and J.~K. Wolf, ``The capacity region of a multiple-access
  discrete memoryless channel can increase with feedback,'' {\em IEEE
  Transactions on Information Theory}, Jan. 1975.

\bibitem{Ozarow:it}
L.~H. Ozarow, ``The capacity of the white {G}aussian multiple access channel
  with feedback,'' {\em IEEE Transactions on Information Theory}, July 1984.

\bibitem{Salman:allterton07}
S.~Avestimehr, S.~Diggavi, and D.~Tse, ``A deterministic approach to wireless
  relay networks,'' {\em Proceedings of Allerton Conference on Communication,
  Control, and computing.}, Sept. 2007.

\bibitem{Alamouti:jsac98}
S.~M. Alamouti, ``A simple transmit diversity technique for wireless
  communication,'' {\em IEEE Journal on Selected Areas in Communications},
  vol.~16, pp.~1451--1458, Oct. 1998.

\bibitem{dtse:it07}
R.~Etkin, D.~Tse, and H.~Wang, ``Gaussian interference channel capacity to
  within one bit,'' {\em IEEE Transactions on Information Theory}, vol.~54,
  pp.~5534--5562, Dec. 2008.

\bibitem{Cover:it79}
T.~M. Cover and A.~A. El-Gamal, ``Capacity theorems for the relay channel,''
  {\em IEEE Transactions on Information Theory}, vol.~25, pp.~572--584, Sept.
  1979.

\bibitem{Cover:it81}
T.~M. Cover and C.~S.~K. Leung, ``An achievable rate region for the
  multiple-access channel with feedback,'' {\em IEEE Transactions on
  Information Theory}, vol.~27, pp.~292--298, May 1981.

\bibitem{Kuhlmann:it89}
F.~Kuhlmann, C.~M. Zeng, and A.~Buzo, ``Achievability proof of some multiuser
  channel coding theorems using backward decoding,'' {\em IEEE Transactions on
  Information Theory}, vol.~35, pp.~1160--1165, Nov. 1989.

\bibitem{HanKoba:it81}
T.~S. Han and K.~Kobayashi, ``A new achievable rate region for the interference
  channel,'' {\em IEEE Transactions on Information Theory}, vol.~27,
  pp.~49--60, Jan. 1981.

\bibitem{ElGamal:it82}
A.~El-Gamal and M.~H. Costa, ``The capacity region of a class of deterministic
  interference channels,'' {\em IEEE Transactions on Information Theory},
  vol.~28, pp.~343--346, Mar. 1982.

\bibitem{Teletar:isit07}
E.~Telatar and D.~Tse, ``Bounds on the capacity region of a class of
  interference channels,'' {\em IEEE International Symposium on Information
  Theory}, June 2007.

\bibitem{Kramer:it02}
G.~Kramer, ``Feedback strategies for white {G}aussian interference networks,''
  {\em IEEE Transactions on Information Theory}, vol.~48, pp.~1423--1438, June
  2002.

\bibitem{Kramer:it04}
G.~Kramer, ``Correction to ``{F}eedback strategies for white {G}aussian
  interference networks'', and a capacity theorem for {G}aussian interference
  channels with feedback,'' {\em IEEE Transactions on Information Theory},
  vol.~50, June 2004.

\bibitem{GastparKramer:06}
M.~Gastpar and G.~Kramer, ``On noisy feedback for interference channels,'' {\em
  In Proc. Asilomar Conference on Signals, Systems, and Computers}, Oct. 2006.

\bibitem{Jiang:07}
J.~Jiang, Y.~Xin, and H.~K. Garg, ``Discrete memoryless interference channels
  with feedback,'' {\em CISS 41st Annual Conference}, pp.~581--584, Mar. 2007.

\bibitem{Vinod:arix09}
V.~Prabhakaran and P.~Viswanath, ``Interference channels with source
  cooperation,'' {\em submitted to IEEE Transaction on Information Theory
  (online available at arXiv:0905.3109v1)}, May 2009.

\bibitem{bresler:europe}
G.~Bresler and D.~Tse, ``The two-user {G}aussian interference channel: a
  deterministic view,'' {\em European Transactions on Telecommunications}, June
  2008.

\bibitem{ahlswede:it}
R.~Ahlswede, N.~Cai, S.-Y.~R. Li, and R.~W. Yeung, ``Network information
  flow,'' {\em IEEE Transactions on Information Theory}, vol.~46,
  pp.~1204--1216, July 2000.

\bibitem{Wu:05}
Y.~Wu, P.~A. Chou, and S.~Y. Kung, ``Information exchange in wireless networks
  with network coding and physical-layer broadcast,'' {\em CISS 39th Annual
  Conference}, Mar. 2005.

\bibitem{Katti:SIGCOMM06}
S.~Katti, H.~Rahul, W.~Hu, D.~Katabi, M.~Medard, and J.~Crowcroft, ``{XOR}s in
  the air: Practical wireless network coding,'' {\em ACM SIGCOMM Computer
  Communication Review}, vol.~36, pp.~243--254, Oct. 2006.

\bibitem{Tassiulas:NetCom09}
L.~Georgiadis and L.~Tassiulas, ``Broadcast erasure channel with feedback -
  capacity and algorithms,'' {\em 2009 Workshop on Network Coding, Theory and
  Applications (NetCod)}, June 2009.

\bibitem{Laneman:it03}
J.~N. Laneman and G.~W. Wornell, ``Distributed space-time-coded protocols for
  exploiting cooperative diversity in wireless networks,'' {\em IEEE
  Transactions on Information Theory}, vol.~49, pp.~2415--2425, Oct. 2003.

\bibitem{SuhTse:arix09}
C.~Suh and D.~Tse, ``Feedback capacity of the {G}aussian interference channel
  to within 1.7075 bits: the symmetric case,'' {\em arXiv:0901.3580v1}, Jan.
  2009.

\bibitem{Tuninetti:isit07}
D.~Tuninetti, ``On {I}nter{F}erence {C}hannel with {G}eneralized {F}eedback
  ({IFC-GF}),'' {\em IEEE International Symposium on Information Theory}, June
  2007.

\bibitem{SK:it}
J.~P.~M. Schalkwijk and T.~Kailath, ``A coding scheme for additive noise
  channels with feedback - part {I}: No bandwith constraint,'' {\em IEEE
  Transactions on Information Theory}, Apr. 1966.

\bibitem{CoverThomas}
T.~M. Cover and J.~A. Thomas, {\em Elements of Information Theory}.
\newblock New York Wiley, 2th~ed., July 2006.

\bibitem{Willems:it82}
F.~M.~J. Willems, ``The feedback capacity region of a class of discrete
  memoryless multiple access channels,'' {\em IEEE Transactions on Information
  Theory}, vol.~28, pp.~93--95, Jan. 1982.

\bibitem{Willems:it89}
A.~P. Hekstra and F.~M.~J. Willems, ``Dependence balance bounds for
  single-output two-way channels,'' {\em IEEE Transactions on Information
  Theory}, vol.~35, pp.~44--53, Jan. 1989.

\end{thebibliography}

\end{document}